\documentclass[a4paper,10pt]{article}
\usepackage{amsthm,amsmath,amssymb,amsfonts}
\usepackage{color}
\usepackage{graphicx,subfig}
\def\Var{\mathrm{Var}}

\def\argmin{\mathrm{argmin}}

\newtheorem{Def}{Definition}
\newtheorem{Prop}{Proposition}

\newtheorem{Cor}{Corollary}
\newtheorem{Rem}{Remark}
\renewenvironment{proof}{\noindent{\bf Proof.}}{\hfill
  $\square$\par\noindent
  }

\newcommand{\norm}[1]{\ensuremath{\vert\!\vert #1 \vert\!\vert}}
\newcommand{\E}{\ensuremath{\mathbb{E}}}

\renewcommand{\P}{\ensuremath{\mathbb{P}}}

\newcommand{\R}{\ensuremath{\mathbb{R}}}
\newcommand{\N}{\ensuremath{\mathbb{N}}}

\newcommand{\e}{\ensuremath{\varepsilon}}

\newcommand{\p}{\ensuremath{\varphi}}

\newcommand{\bbeta}{\text{\mathversion{bold}{$\beta$}}}

\newcommand{\etab}{\text{\mathversion{bold}{$\eta$}}}
\newcommand{\bDelta}{\text{\mathversion{bold}{$\Delta$}}}

\newcommand{\bY}{{\bf Y}}

\newcommand{\bA}{{\bf A}}
\newcommand{\bD}{{\bf D}}
\newcommand{\bG}{{\bf G}}
\newcommand{\bC}{{\bf C}}
\newcommand{\bI}{{\bf I}}
\newcommand{\bx}{{\bf x}}
\newcommand{\be}{{\bf e}}

\newcommand{\by}{{\bf y}}
\newcommand{\J}{{\mathcal J}}
\newcommand{\zero}{\text{\mathversion{bold}{$0$}}}
\newtheorem{lem}{Lemma}
\newtheorem{thm}{Theorem}

\numberwithin{equation}{section}

\begin{document}
\title{Adaptive Lasso and group-Lasso for functional Poisson regression}
\author{S. Ivanoff$^\ddag$, F. Picard$^\star$ \& V. Rivoirard$^{\ddag}$\\  \\ 
  \small{$^\ddag$CEREMADE UMR CNRS 7534, Universit\'e Paris Dauphine,F-75775 Paris, France} \\
  \small{$^\star$LBBE, UMR CNRS 5558 Univ. Lyon 1, F-69622 Villeurbanne, France}  
}
\maketitle

\section*{Abstract}
High dimensional Poisson regression has become a standard framework for the analysis of massive counts datasets. In this work we estimate the intensity function of the Poisson regression model by using a dictionary approach, which generalizes the classical basis approach, combined with a Lasso or a group-Lasso procedure. Selection depends on penalty weights that need to be calibrated. Standard methodologies developed in the Gaussian framework can not be directly applied to Poisson models due to heteroscedasticity. Here we provide data-driven weights for the Lasso and the group-Lasso derived from concentration inequalities adapted to the Poisson case. We show that the associated Lasso and group-Lasso procedures are theoretically optimal in the oracle approach. Simulations are used to assess the empirical performance of our procedure, and an original application to the analysis of Next Generation Sequencing data is provided.

\section*{Introduction}

Poisson functional regression has become a standard framework for image or spectra analysis, in which case observations are made of $n$ independent couples $(Y_i,X_i)_{i=1,\ldots,n}$, and can be modeled as
\begin{equation}\label{modelpoi} 
Y_i|X_i\sim{\mathcal Poisson}(f_0(X_i)). 
\end{equation}
The $X_i$'s (random or fixed) are supposed to lie in a known compact support of $\R^d$ ($d\geq 1$), say $[0,1]^d$, and the purpose is to estimate the unknown intensity function $f_0$ assumed to be positive. Wavelets have been used extensively for intensity estimation, and the statistical challenge has been to propose thresholding procedures in the spirit of \cite{DJ94}, that were adapted to the variance's spatial variability associated with the Poisson framework. An early method to deal with high dimensional count data has been to apply a variance stabilizing-transform (see \cite{anscombe}) and to treat the transformed data as if they were Gaussian. More recently, the same idea has been applied to the data's decomposition in the Haar-wavelet basis, see \cite{fryznason} and \cite{fryzlewicz}, but these methods rely on asymptotic approximations and tend to show lower performance when the level of counts is low \cite{besbeas}. Dedicated wavelet thresholding methods were developed in the Poisson setting by \cite{kolaczyk} and \cite{sardy}, and a recurrent challenge has been to define an appropriate threshold like the universal threshold for shrinkage and selection, as the heteroscedasticity of the model calls for component-wise thresholding.

In this work we first propose to enrich the standard wavelet approach by considering the so-called dictionary strategy. We assume that $\log f_0$ can be well approximated by a linear combination of $p$ known functions, and we reduce the estimation of $f_0$ to the estimation of $p$ coefficients. Dictionaries can be built from classical orthonormal systems such as wavelets, histograms or the Fourier basis, which results in a framework that encompasses wavelet methods. Considering overcomplete (\textit{ie} redundant) dictionaries is efficient to capture different features in the signal, by using sparse representations (see \cite{chen2001} or \cite{Tropp04}). For example, if $\log f_0$ shows piece-wise constant trends along with some periodicity, combining both Haar and Fourier bases will be more powerful than separate strategies, and the model will be sparse in the coefficients domain. To ensure sparse estimations, we consider the Lasso and the group-Lasso procedures. Group estimators are particularly well adapted to the dictionary framework, especially if we consider dictionaries based on a wavelet system, for which it is well known that coefficients can be grouped scale-wise for instance (see \cite{CC05}). Finally, even if we do not make any assumption on $p$ itself, it may be larger than $n$ and methodologies based on $\ell_1$-penalties, such as the Lasso and the group-Lasso appear appropriate.\\

The statistical properties of the Lasso are particularly well understood in the context of regression with \textit{i.i.d.} errors, or for density estimation for which a range of oracle inequalities have been established. These inequalities, now widespread in the literature, provide theoretical error bounds that hold on events with a controllable (large) probability. See for instance \cite{blpv}, \cite{BRT}, \cite{BTW,bunea2} and the references therein. For generalized linear models, \cite{PH07} studied $\ell_1$-regularization path algorithms and \cite{vdG} established non-asymptotic oracle inequalities. The sign consistency of the Lasso has been studied by \cite{JRY} for a very specific Poisson model. Finally, we also mention than the Lasso has also been extensively considered in survival analysis. See for instance \cite{GG},  \cite{Zou2},  \cite{KN13},  \cite{BFJ},  \cite{Lemler} and  \cite{HRBR}. 

Here we consider not only the Lasso estimator but also its extension, the group-Lasso proposed by \cite{YL}, which is relevant when the set of parameters can be partitioned into groups. The analysis of the group-Lasso has been led in different contexts. For instance, consistency has been studied by \cite{Bach08}, \cite{OWJ} and \cite{WH10}. In the linear model, \cite{NR08} derived conditions ensuring various asymptotic properties such as consistency, oracle properties or persistence. Still for the linear model, \cite{LPvdGT} established oracle inequalities and, in the Gaussian setting, pointed out advantages of the group-Lasso with respect to the Lasso, generalizing the results of \cite{CH08} and \cite{HZ10}. We also mention \cite{mvdgb} who studied the group-Lasso for logistic regression, \cite{BLG14} for generalized linear model with Poisson regression as a special case and \cite{DHMS} for other linear heteroscedastic models.\\

As pointed out by empirical comparative studies \cite{besbeas}, the calibration of any thresholding rule is of central importance. Here we consider Lasso and group-Lasso penalties of the form $$\mbox{pen}(\bbeta)=\sum_{j=1}^p \lambda_j |\beta_j|$$ and $$\mbox{pen}^g(\bbeta) =\sum_{k=1}^K \lambda_k^g \|\bbeta_{G_k}\|_2,$$ where $G_1\cup\cdots\cup G_K$ is a partition of $\{1,\ldots,p\}$ into non-overlapping groups (see Section~\ref{sec_methodology} for more details). By calibration we refer to the definition and to the suitable choice of the weights $ \lambda_j$ and $\lambda_k^g$, which is intricate in heteroscedastic models, especially for the group-Lasso. For functional Poissonian regression, the ideal shape of these weights is unknown, even if for the group-Lasso, the $\lambda_k^g$'s should of course depend on the groups size. As for the Lasso, most proposed weights in the literature are non-random and constant such that the penalty is proportional to $\|\bbeta\|_1$, but when facing variable selection and consistency simultaneously, \cite{Zou} showed the interest in considering non-constant data-driven $\ell_1$-weights even in the simple case where the noise is Gaussian with constant variance. This issue becomes even more critical in Poisson functional regression in which variance shows spatial heterogeneity. As \cite{Zou}, our first contribution is to propose here adaptive procedures with weights depending on the data. Weights $\lambda_j$ for the Lasso are derived by using sharp concentration inequalities, in the same spirit as \cite{blpv}, \cite{GG}, \cite{Lemler} and \cite{HRBR}, but adapted to the Poissonian setting. To account for heteroscedasticity, weights $\lambda_j$ are component-specific and depend on the data (see Theorem~\ref{calibration}). We propose a similar procedure for the calibration of the group-Lasso. In most proposed procedures, the analogs of the $\lambda_k^g$'s are proportional to the $\sqrt{|G_k|}$'s (see \cite{NR08},  \cite{bvdg} or \cite{BLG14}). But to the best of our knowledge, adaptive group-Lasso procedures (with weights depending on the data) have not been proposed yet. This is the purpose of Theorem~\ref{theo_1_plus}, which is the main result of this work, generalizing Theorem~\ref{calibration} by using sharp concentration inequalities for infinitely divisible vectors. We show the shape relevance of the data-driven weights $\lambda_k^g$ by comparing them to the weights proposed by \cite{LPvdGT} in the Gaussian framework. In Theorem~\ref{theo_1_plus}, we do not impose any condition on the groups size. However, whether $|G_k|$ is smaller than $\log p$ or not highly influences the order of magnitude of $\lambda_k^g$.

Our second contribution consists in providing the theoretical validity of our approach by establishing slow and fast oracle inequalities under RE-type conditions in the same spirit as \cite{BRT}. Closeness between our estimates and $f_0$ is measured by using the empirical Kullback-Leibler divergence. We show that classical oracle bounds are achieved. We also show the relevance of considering the group-Lasso instead of the Lasso in some situations. Our results, that are non-asymptotic, are valid under very general conditions on the design $(X_i)_{i=1,\ldots,n}$ and on the dictionary. However, to shed some light on our results, we illustrate some of them in the asymptotic setting with classical dictionaries like wavelets, histograms or Fourier bases. Our approach generalizes the classical basis approach and in particular block wavelet thresholding which is equivalent to group-Lasso in that case (see \cite{YL}). We refer the reader to \cite{CC05} for a deep study of block wavelet thresholding in the context of density estimation whose framework shows some similarities with ours in terms of heteroscedasticity. Note that sharp estimation of variance terms proposed in this work can be viewed as an extension of coarse bounds provided by \cite{CC05}. Finally, we emphasize that our procedure differs from \cite{BLG14}'s one in several aspects: First, in their Poisson regression setting, they do not consider a dictionary approach. Furthermore, their weights are constant and not data-driven, so are strongly different from ours. Finally, rates of \cite{BLG14} are established under much more stronger assumptions than ours (see Section~\ref{ssec_ioLasso} for more details). 

Finally, we explore the empirical properties of our calibration procedures by using simulations. We show that our procedures are very easy to implement, and we compare their performance with variance-stabilizing transforms and cross-validation. The calibrated Lasso and group-Lasso are associated with excellent reconstruction properties, even in the case of low counts. We also propose an original application of functional Poisson regression to the analysis of Next Generation Sequencing data, with the search of peaks in Poisson counts associated with the detection of replication origins in the human genome (see \cite{PCA14}). \\

This article is organized as follows. In Section \ref{sec_methodology}, we introduce the Lasso and group-Lasso procedures we propose in the dictionary approach setting. In Section~\ref{sec_calibration}, we derive data-driven weights of our procedures that are extensively commented. Theoretical performance of our estimates are studied in Section \ref{sec_io} in the oracle approach. In Section \ref{sec_simulation}, we investigate the empirical performance of the proposed estimators using simulated data, and an application is provided on next generation sequencing data in Section \ref{sec_applications}.
\section{Penalized log-likelihood estimates for Poisson regression and dictionary approach}\label{sec_methodology}
We consider the functional Poisson regression model, with $n$ observed counts $Y_i \in \N$ modeled such that:
\begin{equation}
Y_i|X_i\sim{\mathcal Poisson}(f_0(X_i)),
\end{equation}
with the $X_i$'s (random or fixed) supposed to lie in a known compact support, say $[0,1]^d$. Since the goal here is to estimate the  function $f_0$ assumed to be positive on $[0,1]^d,$ a natural candidate is a function $f$ of the form $f=\exp(g)$. Then, we consider the so-called dictionary approach which consists in decomposing $g$ as a linear combination of the elements of a given finite dictionary of functions denoted by $\Upsilon=\{\p_j\}_{j \in \J}$, with $\|\p_j\|_2=1$ for all $j$. Consequently, we choose $g$ of the form: $$g = \sum_{j\in \J} \beta_j \p_{j},$$ with $p=\mbox{card}(\J)$ that may depend on $n$ (as well as the elements of $\Upsilon$). Without loss of generality we will assume in the following  that $\J = \{1,\ldots,p\}$. In this framework, estimating $f_0$ is equivalent to selecting the vector of regression coefficients $\bbeta = (\beta_j)_{j\in \J} \in \R^p$. In the sequel, we write $g_\bbeta = \sum_{j\in \J} \beta_j \p_j$, $f_\bbeta = \exp(g_\bbeta),$ for all $\bbeta \in \R^p$. Note that we do not require the model to be true, that is we do not suppose the existence of $\bbeta_0$ such that $f_0 = f_{\bbeta_0}$.

The strength of the dictionary approach lies in its ability to capture different features of the function to estimate (smoothness, sparsity, periodicity,...) by sparse combinations of elements of the dictionary so that only few coefficients need to be selected, which limits estimation errors. Obviously, the dictionary approach encompasses the classical basis approach consisting in decomposing $g$ on an orthonormal system. The richer the dictionary, the sparser the decomposition, so $p$ can be larger than $n$ and the model becomes high-dimensional.   

We consider a likelihood-based penalized criterion to select $\bbeta$, the coefficients of the dictionary decomposition. We denote by
$\bA$ the $n \times p$-design matrix with $A_{ij} = \p_j(X_i)$, $\bY=(Y_1,\ldots,Y_n)^T$ and the log-likelihood associated with this model is 
$$l(\bbeta)=\sum_{j\in \J}\beta_{j}(\bA^T\bY)_{j}-\sum_{i=1}^n\exp\Bigl(\sum_{j\in \J}\beta_{j}A_{ij}\Bigr)-\sum_{i=1}^n\log (Y_i!),$$
which is a concave function of $\bbeta$. Next sections propose two different ways to penalize $-l(\bbeta)$.

\subsection{The Lasso estimate} 
The first penalty we propose is based on the (weighted) $\ell_1$-norm and we obtain a Lasso-type estimate by considering 
\begin{equation}\label{critLasso}
\widehat\bbeta^L\in\underset{\bbeta \in \R^p}{\argmin}\left\{ - l(\bbeta) + \sum_{j=1}^p \lambda_j |\beta_j| \right\}.
\end{equation}
The penalty term $\sum_{j=1}^p \lambda_j |\beta_j|$ depends on positive weights $(\lambda_j)_{j\in \J}$ that vary according to the elements of the dictionary and are chosen in Section~ \ref{ssec_lambdaj}. This choice of varying weights instead of a unique $\lambda$ stems from heteroscedasticity due to the Poisson regression, and a first part of our work consists in providing theoretical data-driven values for these weights, in the same spirit as \cite{blpv} or \cite{HRBR} for instance.  From the first order optimality conditions (see \cite{bvdg}), $\widehat\bbeta^L$ satisfies
\begin{eqnarray*}
\left\{
\begin{aligned}
\bA_j^T(\bY-\exp(\bA\widehat\bbeta^{L})) &=& \lambda_j \frac{\widehat\beta_{j}^L}{|\widehat\beta_{j}^L|} \quad\, \text{  if } \widehat\beta_j^L\neq 0,\\
|\bA_j^T(\bY-\exp(\bA\widehat\bbeta^{L}))| &\leq& \lambda_j \qquad\quad\text{  if } \widehat\beta_j^L=0,
\end{aligned}
\right.
\end{eqnarray*}
where  $\exp(\bA\bbeta)=\left(\exp((\bA\bbeta)_1),\ldots,\exp((\bA\bbeta)_n)\right)^T$ and $\bA_j$ is the $j$-th column of the matrix $\bA$.  Note that the larger the $\lambda_j$'s, the sparser the
estimates. In particular $\widehat\bbeta^{L}$ belongs to the set of the vectors $\bbeta\in\R^p$ that satisfies for any $j\in \J$,
\begin{equation}\label{cond_dantzig}
|\bA_j^T(\bY-\exp(\bA\bbeta))|\leq \lambda_j.
\end{equation}

The Lasso estimator of $f_0$ is now easily derived.
\begin{Def}
The Lasso estimator of $f_0$ is defined as
$$\widehat f^L(x):=\exp(\widehat g^L(x)):=\exp\Biggl(\sum_{j=1}^p\widehat\beta_j^L\p_j(x)\Biggr).$$
\end{Def}
\noindent We also propose an alternative to $\widehat f^L$ by considering the group-Lasso.

\subsection{The group-Lasso estimate}
We also consider the grouping of coefficients into non-overlapping blocks. Indeed, group estimates may be better adapted than their single counterparts when there is a natural group structure. The procedure keeps or discards all the coefficients within a block and can increase estimation accuracy by using information about coefficients of the same block. In our setting, we partition the set of indices $\J = \{1,\ldots,p\}$ into $K$ non-empty groups: 
$$ \{1,\ldots,p\} = G_1 \cup G_2 \cup \cdots \cup G_K.$$
For any $\bbeta \in \R^p$, $\bbeta_{G_k}$ stands for the sub-vector of $\bbeta$ with elements indexed by the elements of $G_k$, and we define the block $\ell_1$-norm on $\R^p$ by
$$\|\bbeta\|_{1,2} = \sum_{k=1}^K \|\bbeta_{G_k}\|_2. $$
Similarly, $\bA_{G_k}$ is the $n \times |G_k|$ submatrix of $\bA$ whose columns are indexed by the elements of $G_k$. Then the group-Lasso $\widehat\bbeta^{gL}$ is a solution to the following convex optimization problem:
$$\widehat\bbeta^{gL} \in \underset{\bbeta \in \R^p}{\argmin} \Big\{-l(\bbeta) + \sum_{k=1}^K \lambda_k^g \|\bbeta_{G_k}\|_2\Big\},$$
where the $\lambda_k^g$'s are positive weights for which we also provide a theoretical data-driven expression in Section~\ref{ssec_lambdak}. This group-estimator is constructed similarly to the Lasso, with the block $\ell_1$-norm being used instead of the $\ell_1$-norm. In particular, note that if all groups are of size one then we recover the Lasso estimator. Convex analysis states that $\widehat\bbeta^{gL}$ is a solution of the above optimization problem if the $p$-dimensional vector $\zero$ is in the subdifferential of the objective function. Therefore, $\widehat\bbeta^{gL}$ satisfies:
\begin{eqnarray*}
\left\{
\begin{aligned}
\bA_{G_k}^T(\bY-\exp(\bA\widehat\bbeta^{gL})) &= \lambda_k^g \frac{\widehat\bbeta_{G_k}^{gL}}{\|\widehat\bbeta_{G_k}^{gL}\|_2} \quad &\text{  if } \widehat\bbeta_{G_k}^{gL}\neq \zero,\\
\|\bA_{G_k}^T(\bY-\exp(\bA\widehat\bbeta^{gL}))\|_2 &\leq \lambda_k^g \qquad\qquad &\text{  if } \widehat\bbeta_{G_k}^{gL}=\zero.
\end{aligned}
\right.
\end{eqnarray*}
This procedure naturally enhances group-sparsity as analyzed by \cite{YL}, \cite{LPvdGT} and references therein. 

Obviously, $\widehat\bbeta^{gL}$ belongs to the set of the vectors $\bbeta\in\R^p$ that satisfy for any $k \in \{1,\ldots,K\},$
\begin{equation}\label{cond_grp_dantzig}
\|\bA_{G_k}^T(\bY-\exp(\bA\bbeta))\|_2 \leq \lambda_k^g.
\end{equation}
Now, we set
\begin{Def}
The group Lasso estimator of $f_0$ is defined as
$$\widehat f^{gL}(x):=\exp(\widehat g^{gL}(x)):=\exp\Biggl(\sum_{j=1}^p\widehat\beta_j^{gL}\p_j(x)\Biggr).$$
\end{Def}

In the following our results are given conditionally on the $X_i$'s, and $\E$ (resp. $\P$) stands for the expectation (resp. the probability measure) conditionally on $X_1,\ldots,X_n$. In some situations, to give orders of magnitudes of some expressions, we will use the following definition:
\begin{Def}\label{def-reg}
We say that the design $(X_i)_{i=1,\ldots,n}$ is regular if either the design is deterministic and the $X_i$'s are equispaced in $[0,1]$ or the design is random and the $X_i$'s are i.i.d. with density $h$, with
$$0<\inf_{x\in[0,1]^d}h(x)\leq  \sup_{x\in[0,1]^d}h(x)<\infty.$$
\end{Def}

\section{Weights calibration using concentration inequalities}\label{sec_calibration}
 Our first contribution is to derive theoretical data-driven values of the weights  $\lambda_j$'s and $\lambda_k^g$'s, specially adapted to the Poisson model. In the classical Gaussian framework with noise variance $\sigma^2$, weights for the Lasso are chosen to be proportional to  $\sigma\sqrt{\log p}$ (see \cite{BRT} for instance). The Poisson setting is more involved due to heteroscedasticity and such simple tuning procedures cannot be generalized easily. Sections~\ref{ssec_lambdaj} and \ref{ssec_lambdak} give closed forms of parameters $\lambda_j$ and $\lambda_k^g$. They are based on  concentration inequalities specific to the Poisson model. In particular, $\lambda_j$ is used to control the fluctuations of $\bA_j^T\bY$ around its mean, which enhances the key role of $V_j$, a variance term (the analog of $\sigma^2$) defined by
\begin{equation}\label{def_vj}
V_j = \Var (\bA_j^T\bY) = \sum_{i=1}^n f_0(X_i)\p_j^2(X_i).
\end{equation}
\subsection{Data-driven weights for the Lasso procedure}\label{ssec_lambdaj}
For any $j$, we choose a data-driven value for $\lambda_j$ as small as possible so that with high probability, for any $j\in \J$, 
\begin{equation}\label{conc_dantzig}
|\bA_j^T(\bY-\E[\bY])|\leq\lambda_j.
\end{equation}
Such a control is classical for Lasso estimates (see the references above) and is also a key point of the technical arguments of the proofs. Requiring that the weights are as small as possible is justified, from the theoretical point of view, by oracle bounds depending on the $\lambda_j$'s (see Corollaries \ref{cor-simple} and \ref{cor-RE}).  Furthermore, as discussed in \cite{blpv}, choosing theoretical Lasso weights as small as possible is also a suitable guideline for practical purposes.  Finally, note that if the model were true, \textit{i.e.} if there existed a true sparse vector $\bbeta_0$ such that $f_0 = f_{\bbeta_0}$, then $\E[\bY]=\exp(\bA\bbeta_0)$ and $\bbeta_0$ would belong to the set defined by 
\eqref{cond_dantzig} with large probability. The smaller the $\lambda_j$'s, the smaller the set
within selection of $\widehat\bbeta^L$ is performed. So, with a sharp control in \eqref{conc_dantzig}, we increase the probability to select $\bbeta_0$. The following theorem provides the data-driven weights $\lambda_j$'s. The main theoretical ingredient we use to choose the weights $\lambda_j$'s is a concentration inequality for Poisson processes and to proceed, we link the quantity $\bA_j^T\bY$ to a specific Poisson process, as detailed in the proofs Section~\ref{sec:proof1}.
\begin{thm}\label{calibration}
Let $j$ be fixed and $\gamma >0$ be a constant. Define $\widehat V_j = \sum_{i=1}^n \p_j^2(X_i)Y_i$ the natural unbiased estimator of $V_j$ and
 $$\widetilde V_j=\widehat V_j+\sqrt{2\gamma\log p\widehat V_j\max_i \p_j^2(X_i)}+3\gamma\log p\max_i \p_j^2(X_i).$$
Set
\begin{equation}\label{choix_de_lambda}
 \lambda_j=\sqrt{2\gamma\log p\widetilde V_j}+\frac{\gamma\log p}{3}\max_i |\p_j(X_i)|,
\end{equation}
then
\begin{equation}\label{controle_via_lambda}
\P\Big(|\bA_j^T(\bY-\E[\bY])| \geq \lambda_j \Big) \leq \frac{3}{p^\gamma}.
\end{equation}
\end{thm}
The first term $\sqrt{2\gamma\log p\widetilde V_j}$ in $\lambda_j$ is the main one, and constitutes a variance term depending on $\widetilde V_j$ that slightly overestimates $V_j$ (see Section~\ref{sec:proof1} for more details about the derivation of $\widetilde V_j$). Its dependence on an estimate of $V_j$ was expected since we aim at controlling fluctuations of $\bA_j^T\bY$ around its mean. The second term comes from the heavy tail of the Poisson distribution, and is the price to pay, in the non-asymptotic setting, for the added complexity of the Poisson framework compared to the Gaussian framework. 

To shed more lights on the form of the proposed weights from the asymptotic point of view, assume that the design is regular (see Definition~\ref{def-reg}). In this case, it is easy to see that under mild assumptions on $f_0$, $V_j$ is asymptotically of order $n$. If we further assume that  
\begin{equation}\label{hyp}
\max_i |\p_j(X_i)|=o(\sqrt{n/\log p}),
\end{equation}
then, when $p$ is large, with high probability, $\widehat V_j$ (and then $\widetilde V_j$) is also of order $n$ (using Remark \ref{order} in the proofs Section~\ref{sec:proof1}), and the second term in $\lambda_j$ is negligible with respect to the first one. In this case, $\lambda_j$ is of order $\sqrt{n\log p}.$ Note that Assumption \eqref{hyp} is quite classical in heteroscedastic settings (see \cite{blpv}). By taking the hyperparameter $\gamma$ larger than 1,  then for large values of $p$, \eqref{conc_dantzig} is true for any $j\in {\mathcal J}$, with large probability.

\subsection{Data-driven weights for the group Lasso procedure}\label{ssec_lambdak}
 Current group-Lasso procedures are tuned by choosing the analog of $\lambda_k^g$ proportional to $\sqrt{|G_k|}$ (see \cite{NR08}, Chapter 4 of \cite{bvdg} or \cite{BLG14}). A more refined version of tuning group-Lasso is provided by \cite{LPvdGT} in the Gaussian setting (see below for a detailed discussion). To the best of our knowledge, data-driven weights (with theoretical validation) for the group-Lasso have not been proposed yet. It is the purpose of Theorem~\ref{theo_1_plus}. Similarly to the previous section, we propose data-driven theoretical derivations for the weights $\lambda_k^g$'s that are chosen as small as possible, but satisfying for any $k\in\{1,\ldots,K\}$,
\begin{equation}\label{conc_gdant}
\|\bA_{G_k}^T(\bY-\E[\bY])\|_2\leq \lambda_k^g
\end{equation}
with high probability (see \eqref{cond_grp_dantzig}). Choosing the smallest possible weights is also recommended by \cite{LPvdGT} in the Gaussian setting (see in their Section~3 the discussion about weights and comparisons with coarser weights of \cite{NR08}). Obviously, $\lambda_k^g$ should depend on sharp estimates of the variance parameters $(V_j)_{j\in G_k}$. The following theorem is the equivalent of Theorem~\ref{calibration} for the group-Lasso. Relying on specific concentration inequalities established for infinitely divisible vectors by \cite{HMRB}, it requires a known upper bound for $f_0$, which can be chosen as $\underset{i}{\max} \,Y_i$ in practice. 
\begin{thm}\label{theo_1_plus}
  Let $k\in\{1,\ldots,K\}$ be fixed and $\gamma >0$ be a constant. Assume that there exists $M>0$ such that for any $x$, $|f_0(x)|\leq M$. 
  Let 
  \begin{equation}\label{require_c}
    c_k=\sup_{\bx \in \R^n}\frac{\|\bA_{G_k}\bA_{G_k}^T\bx\|_2}{\|\bA_{G_k}^T\bx\|_2}.
  \end{equation}
  For all $j \in G_k$, still with $\widehat V_j = \sum_{i=1}^n \p_j^2(X_i)Y_i,$ define
  \begin{equation}\label{def_vj_2}
    \widetilde V_j^g=\widehat V_j+\sqrt{2(\gamma\log p + \log |G_k|)\widehat V_j\max_i \p_j^2(X_i)}+3(\gamma\log p + \log |G_k|)\max_i \p_j^2(X_i).
  \end{equation}
  Let $\gamma>0$ be fixed. Define ${b_k^i} = \sqrt{\sum_{j\in G_k} \p_j^2(X_i)}$ and $b_k = \underset{i}{\max}\, {b_k^i}$. Finally, we set
  \begin{equation}\label{def_lambda_k^g}
    \lambda_k^g = \left(1+\frac{1}{2\sqrt{2\gamma\log p}}\right)\sqrt{\sum_{j\in G_k} \widetilde V_j^g} + 2 \sqrt{\gamma\log p\, D_k},
  \end{equation}
  where $D_k = 8Mc_k^2 + 16b_k^2\gamma\log p$. Then,
  \begin{equation}\label{theoreme_plus}
    \P\Bigg(\|\bA_{G_k}^T(\bY-\E[\bY])\|_2 \geq \lambda_k^g \Bigg) \leq \frac{2}{p^\gamma}.
  \end{equation}
\end{thm}
Similarly to the weights $\lambda_j$'s of the Lasso, each weight $\lambda_k^g$ is the sum of two terms. The term $\widetilde V_j^g$ is an estimate of $V_j$ so it plays the same role as $\widetilde V_j$. In particular,  $\widetilde V_j^g$ and $\widetilde V_j$ are of the same order since $\log |G_k|$ is not larger than $\log p$. The first term in $\lambda_k^g$ is a variance term, and the leading constant $1+1/(2\sqrt{2\gamma\log p})$ is close to 1 when $p$ is large. So, the first term is close to the square root of the sum of sharp estimates of the $(V_j)_{j\in G_k}$, as expected for a grouping strategy (see \cite{CC05}).

The second term, namely  $2 \sqrt{\gamma\log p\, D_k}$, is more involved. To shed light on it, since $b_k$ and $c_k$ play a key role, we first state the following proposition controlling values of these terms.
\begin{Prop}\label{prop}
Let $k$ be fixed. We have
\begin{equation}\label{c_inf_b}
b_k\leq c_k\leq \sqrt{n}b_k.
\end{equation}
Furthermore,
 \begin{equation}\label{majock}
 c_k^2 \leq \max_{j \in G_k} \sum_{j' \in G_k} \Big|\sum_{l=1}^n \p_j(X_l)\p_{j'}(X_l)\Big|.
 \end{equation}
\end{Prop}

The first inequality of Proposition~\ref{prop} shows that  $2 \sqrt{\gamma\log p\, D_k}$ is smaller than $c_k\sqrt{\log p}+b_k\log p\leq 2c_k\log p$ up to a constant depending on $\gamma$ and $M$. At first glance, the second inequality of Proposition~\ref{prop} shows that $c_k$ is controlled by the coherence of the dictionary (see \cite{Tropp04}) and $b_k$ depends on $(\max_i|\p_j(X_i)|)_{j\in G_k}$. In particular, if for a given block $G_k$, the functions $(\p_j)_{j\in G_k}$ are orthonormal, then for fixed $j\not=j'$,  if the $X_i$'s are deterministic and equispaced on $[0,1]$ or if the $X_i$'s are i.i.d. with a uniform density on $[0,1]^d$, then, when $n$ is large
$$\frac{1}{n}\sum_{l=1}^n \p_j(X_l)\p_{j'}(X_l)\approx \int \p_j(x)\p_{j'}(x)dx=0$$
and we expect
$$c_k^2\lesssim \max_{j \in G_k}\sum_{l=1}^n \p_j^2(X_l).$$
In any case, by using the Cauchy-Schwarz Inequality,   Condition \eqref{majock} gives
\begin{equation}\label{majobrute}
c_k^2 \leq \max_{j \in G_k} \sum_{j' \in G_k} \left(\sum_{l=1}^n \p_j^2(X_l)\right)^{1/2} \left(\sum_{l=1}^n\p_{j'}^2(X_l)\right)^{1/2}.
\end{equation}
To further discuss orders of magnitude for the $c_k$'s, we consider the following condition
\begin{equation}\label{asyn}
 \max_{j \in G_k}\sum_{l=1}^n \p_j^2(X_l)=O(n),
\end{equation}
which is satisfied for instance for fixed $k$ if the design is regular, since $\|\p_j\|_2=1$. Under Assumption \eqref{asyn}, Inequality \eqref{majobrute} gives $$c_k^2=O(|G_k|n).$$ We can say more on $b_k$ and $c_k$ (and then on the order of magnitude of $\lambda_k^g$) by considering classical dictionaries of the literature to build the blocks $G_k$, which is of course realized in practice. In the subsequent discussions, the balance between $|G_k|$ and $\log p$ plays a key role. Note also that $\log p$ is the group size often recommended in the classical setting ($p=n$) for block thresholding (see Theorem 1 of \cite{CC05}).

\subsubsection{Order of magnitude of $\lambda_k^g$ by considering classical dictionaries.}\label{sec:order}
Let $G_k$ be a given block and assume that it is built by using only one of the subsequent systems. For each example, we discuss the order of magnitude of the term $D_k= 8Mc_k^2 + 16b_k^2\gamma\log p$. For ease of exposition, we assume that $f_0$ is supported by $[0,1]$ but we could easily generalize the following discussion to the multidimensional setting.  

\paragraph{Bounded dictionary.} Similarly to \cite{BLG14}, we assume that there exists a constant $L$ not depending on $n$ and $p$ such that for any $j\in G_k$, $\|\p_j\|_\infty\leq L.$ For instance, atoms of the Fourier basis satisfy this property. We then have $$b_k^2\leq L^2|G_k|.$$ Finally, under Assumption \eqref{asyn},
\begin{equation}\label{Dk-Four}
D_k=O(|G_k|n+|G_k|\log p).
\end{equation}

\paragraph{Compactly supported wavelets.} Consider the one-dimensional Haar dictionary: For $j=(j_1,k_1)\in{\mathbb Z}^2$ we set $\p_j(x)=2^{j_1/2}\psi(2^{j_1}x-k_1),\quad \psi(x)=1_{[0,0.5]}(x)-1_{]0.5,1]}(x).$ Assume that the block $G_k$ depends on only one resolution level $j_1$: $G_k=\{j=(j_1,k_1): \quad k_1\in B_{j_1}\},$  where $B_{j_1}$ is a subset of $\{0,1,\ldots, 2^{j_1}-1\}$. In this case, since for $j,j'\in G_k$ with $j\not=j'$, for any $x$, $\p_j(x)\p_{j'}(x)=0$, $$b_k^2=\max_i\sum_{j\in G_k} \p_j^2(X_i)=\max_{i,j\in G_k}\p_j^2(X_i)= 2^{j_1}$$ and Inequality \eqref{majock} gives $$c_k^2\leq \max_{j \in G_k} \sum_{l=1}^n \p_j^2(X_l).$$ If, similarly to Condition \eqref{hyp}, we assume that $\max_{i,j\in G_k} |\p_j(X_i)|=o(\sqrt{n/\log p}),$ then $$b_k^2=o(n/\log p),$$ and under Assumption \eqref{asyn}, $$D_k=O(n),$$ which improves \eqref{Dk-Four}. This property can be easily extended to general compactly supported wavelets $\psi$, since, in this case, for any $j=(j_1,k_1)$$$S_j=\left\{j'=(j_1,k'_1):\quad k'_1\in{\mathbb Z}, \ \p_j\times\p_{j'}\not\equiv 0\right\}$$ is finite with cardinal only depending on the support of $\psi$.

\paragraph{Regular histograms.} Consider a regular grid of the interval $[0,1]$, $\{0,\delta,2\delta,\ldots\}$ with $\delta>0$. Consider then $(\p_j)_{j\in G_k}$ such that for any $j\in G_k$, there exists $\ell$ such that $\p_j=\delta^{-1/2}1_{(\delta(\ell-1), \delta\ell]}.$ We have $\|\p_j\|_2=1$ and $\|\p_j\|_\infty=\delta^{-1/2}.$ As for the wavelet case, for $j,j'\in G_k$ with $j\not=j'$, for any $x$, $\p_j(x)\p_{j'}(x)=0$, then  $$b_k^2=\max_i\sum_{j\in G_k} \p_j^2(X_i)=\max_{i,j\in G_k}\p_j^2(X_i)=\delta^{-1}.$$ If, similarly to Condition \eqref{hyp}, we assume that $\max_{i,j\in G_k} |\p_j(X_i)|=o(\sqrt{n/\log p}),$ then  $$b_k^2=o(n/\log p),$$ and under Assumption~\eqref{asyn}, $$D_k=O(n).$$

The previous discussion shows that we can exhibit dictionaries such that $c_k^2$ and $D_k$ are of order $n$ and the term $b_k^2\log p$ is negligible with respect to $c_k^2$. Then, if similarly to Section~\ref{ssec_lambdaj}, the terms $(\widetilde V_j^g)_{j\in G_k}$ are all of order $n$, $\lambda_k^g$ is of order $\sqrt{n\times\max(\log p; |G_k|)}$ and the main term in $\lambda_k^g$ is the first one as soon as $|G_k|\geq \log p$. In this case, $\lambda_k^g$ is of order $\sqrt{|G_k|n}$.

\subsubsection{Comparison with the Gaussian framework.}
Now, let us compare the $\lambda_k^g$'s to the weights proposed by \cite{LPvdGT} in the Gaussian framework. Adapting their notations to ours, \cite{LPvdGT} estimate the vector $\bbeta_0$ in the model $\bY \sim {\mathcal N}(\bA\bbeta_0, \sigma^2 \bI_n)$ by using the group-Lasso estimate with weights equal to $$\widetilde{\lambda}_k^{g}=2\sqrt{\sigma^2\Big(Tr(\bA_{G_k}^T\bA_{G_k}) + 2|||\bA_{G_k}^T\bA_{G_k}|||(2\gamma\log p + \sqrt{|G_k|\gamma\log p})\Big),}$$ where $|||\bA_{G_k}^T\bA_{G_k}|||$ denotes the maximal eigenvalue of $\bA_{G_k}^T\bA_{G_k}$ (see (3.1) in \cite{LPvdGT}). So, if $|G_k|\leq \log p$, the above expression is of the same order as
\begin{equation}\label{majoSacha}
\sqrt{\sigma^2Tr(\bA_{G_k}^T\bA_{G_k}) }+\sqrt{\sigma^2|||\bA_{G_k}^T\bA_{G_k}|||\gamma\log p}.
\end{equation}
Neglecting the term $16b_k^2\gamma\log p$ in the definition of $D_k$ (see the discussion in Section~\ref{sec:order}), we observe that $\lambda_k^g$ is of the same order as
\begin{equation}\label{contr�le poids}
 \sqrt{\sum_{j\in G_k}  \widetilde V_j^g} + \sqrt{Mc_k^2\gamma \log p}.
\end{equation}
Since $M$ is an upper bound of $\Var(Y_i)=f_0(X_i)$ for any $i$, strong similarities can be highlighted between the forms of the weights in the Poisson and Gaussian settings:
\begin{itemize}
\item[-] For the first terms, $ \widetilde V_j^g$ is an estimate of $V_j$  and
$$\sum_{j\in G_k}V_j\leq M\sum_{j\in G_k}\sum_{i=1}^n \p_j^2(X_i)=M\times Tr(\bA_{G_k}^T\bA_{G_k}).$$ 
\item[-] For the second terms, in view of \eqref{require_c},  $c_k^2$ is related to $|||\bA_{G_k}^T\bA_{G_k}|||$ since we have
$$
c_k^2=\sup_{\bx\in\R^n} \frac{\|\bA_{G_k}\bA_{G_k}^T\bx\|_2^2}{\|\bA_{G_k}^T\bx\|_2^2} \leq \sup_{\by\in\R^{|G_k|}} \frac{\|\bA_{G_k}\by\|_2^2}{\|\by\|_2^2}= |||\bA_{G_k}^T\bA_{G_k}|||.$$
\end{itemize}
These strong similarities between the Gaussian and the Poissonian settings strongly support the shape relevance of the weights we propose.

\subsubsection{Suboptimality of the naive procedure}
Finally, we show that the naive procedure that considers $ \sqrt{\sum_{j\in G_k} \lambda_j^2}$ instead of  $\lambda_k^g$ is suboptimal even if, obviously due to Theorem~\ref{calibration}, with high probability,
$$ \|\bA_{G_k}^T(\bY-\E[\bY])\|_2\leq \sqrt{\sum_{j\in G_k} \lambda_j^2}.$$ Suboptimality is justified by following heuristic arguments. Assume that for all $j$ and $k$, the first terms in \eqref{choix_de_lambda} and \eqref{def_lambda_k^g} are the main ones and $\widetilde V_j\approx \widetilde V_j^g\approx V_j$. Then by considering  $\lambda_k^g$ instead of $\sqrt{\sum_{j\in G_k} \lambda_j^2}$, we improve our weights by the factor $\sqrt{\log p}$, since in this situation, $$\lambda_k^g \approx \sqrt{\sum_{j\in G_k}V_j}$$ and $$\sqrt{\sum_{j \in G_k} \lambda_j^2 } \approx \sqrt{\log p \sum_{j \in G_k}V_j}\approx \sqrt{\log p}\,\lambda_k^g.$$ Remember that our previous discussion shows the importance to consider weights as small as possible as soon as \eqref{conc_gdant} is satisfied with high probability. The next section will confirm this point.

\section{Oracle inequalities}\label{sec_io}
In this section, we establish oracle inequalities to study theoretical properties of our estimation procedures. The $X_i$'s are still assumption-free,  and the performance of our procedures will be only evaluated at the $X_i'$s. To measure the closeness between $f_0$ and an estimate, we use the empirical Kullback-Leibler divergence associated with our model, denoted by $K(\cdot,\cdot)$. Straightforward computations (see for instance \cite{LL06}) show that for any positive function $f$, 
\begin{eqnarray*}
K(f_0,f) &=& \E\left[\log\left(\frac{{\mathcal L}(f_0)}{{\mathcal L}(f)}\right)\right]\\
&=&\sum_{i=1}^n\left[(f_0(X_i)\log f_0(X_i)-f_0(X_i))\right]- \left[(f_0(X_i)\log f(X_i)-f(X_i))\right],
\end{eqnarray*}
where ${\mathcal L}(f)$ is the likelihood associated with $f$. We speak about {\it empirical divergence} to emphasize its dependence on the $X_i$'s. Note that we can write 
\begin{equation}\label{k_pos}
K(f_0,f) = \sum_{i=1}^n f_0(X_i) (e^{u_i}-u_i-1), 
\end{equation}
where $u_i = \log {f(X_i) \over f_0(X_i)}.$ This expression clearly shows that $K(f_0,f)$ is non-negative and $K(f_0,f)=0$ if and only if for all $i\in\{1,\ldots,n\}$, we have $u_i = 0$, that is $f(X_i) = f_0(X_i)$ for all $i\in\{1,\ldots,n\}$. 

\begin{Rem}
To weaken the dependence on $n$ in the asymptotic setting, an alternative, not considered here, would consist in considering $n^{-1}K(\cdot,\cdot)$ instead of  $K(\cdot,\cdot)$.
\end{Rem}
If the classical ${\mathbb L}_2$-norm is the natural loss-function for penalized least squares criteria, the empirical Kullback-Leibler divergence is a natural alternative for penalized likelihood criteria. In next sections, oracle inequalities will be expressed by using $K(\cdot,\cdot).$

\subsection{Oracle inequalities for the group-Lasso estimate}\label{ssec_ioLasso}
In this section, we state oracle inequalities for the group-Lasso. These results can be viewed as generalizations of results by \cite{LPvdGT} to the case of the Poisson regression model. They will be established on the set $\Omega_g$ where 
\begin{equation}\label{def_omega}
\Omega_g = \Big\{\|\bA_{G_k}^T(\bY-\E[\bY])\|_2 \leq \lambda_k^g\quad\forall\, k\in \{1,\ldots,K\}\Big\}.
\end{equation}
Under assumptions of Theorem \ref{theo_1_plus}, we have $\P(\Omega_g) \geq 1-{2K\over p^\gamma}\geq 1-2p^{1-\gamma}$. By considering $\gamma>1$, we have that $\P(\Omega_g)$ goes to 1 at a polynomial rate of convergence when $p$ goes to $+\infty$. For any $\bbeta\in\R^p$, we denote by $$f_\bbeta(x)=\exp\Biggl(\sum_{j=1}^p\beta_j\p_j(x)\Biggr),$$ the candidate associated with $\bbeta$ to estimate $f_0$. We first give a {\it slow oracle inequality} (see for instance \cite{BTW}, \cite{GG} or \cite{LPvdGT}) that does not require any assumption.
\begin{thm}\label{slow_io_gl}
On $\Omega_g$, 
\begin{equation}
\label{slow_oi_gLasso}
K(f_0,\widehat f^{gL}) \leq \inf_{\bbeta \in \R^p} \Big\{K(f_0, f_\beta) + 2 \sum_{k=1}^K \lambda_k^g\|\bbeta_{G_k}\|_2\Big\}.
\end{equation}
\end{thm}
Note that $$\sum_{k=1}^K \lambda_k^g\|\bbeta_{G_k}\|_2\leq \max_{k\in\{1,\ldots,K\}}\lambda_{k}^{g}\times\|\bbeta\|_{1,2}$$ and  \eqref{slow_oi_gLasso}  is then similar to Inequality (3.9) of \cite{LPvdGT}. We can improve the rate of \eqref{slow_oi_gLasso} at the price of stronger assumptions on the matrix $\bA$. We consider the following assumptions:

\medskip

\noindent\textbf{Assumption 1.} There exists $\mu > 0$ such that the convex set $$\Gamma(\mu) = \left\{\bbeta \in \R^p:\quad\max_{i\in\{1,\ldots,n\}} \left|\sum_{j=1}^p \beta_j\p_j(X_i) - \log f_0(X_i)\right| \leq \mu\right\}$$
contains a non-empty open set of $\R^p$. 

\medskip

In the sequel, we restrict our attention to estimates $\widehat\bbeta^{gL}$ belonging to $\Gamma(\mu)$. Note that we do not impose any upper bound on $\mu$ so this assumption is quite mild. This assumption (or variations of it) has already been considered by \cite{vdG}, \cite{KN13} and \cite{Lemler}. Its role consists in connecting $K(.,.)$ to some empirical quadratic loss functions (see the proof of Theorem \ref{fast_io_gl}).

\medskip

\noindent\textbf{Assumption 2}. For some integer $s \in \{1,\ldots,K\}$ and some constant $r$, the following condition holds:
$$0 < \kappa_n(s,r) = \min_{\substack{J\subset \{1,\ldots,K\}\\|J|\leq s}} \min_{\substack{\bbeta\in\R^p-\{0\}\\\|\bbeta_{J^c}\|_{1,2}\leq r\|\bbeta_J\|_{1,2}}} \frac{(\bbeta^T\bG\bbeta)^{1/2}}{\|\bbeta_J\|_2},$$
where $\bG$ is the Gram matrix defined by $\bG = \bA^T\bC\bA$, where $\bC$ is the diagonal matrix with $C_{i,i}=f_0(X_i)$. 
With a slight abuse, $\bbeta_{J}$  (resp. $\bbeta_{J^c}$) stands for the sub-vector of $\bbeta$ with elements indexed by the indices of the groups $(G_{k})_{k\in J}$ (resp. $(G_{k})_{k\in J^c}$).

\medskip

This assumption is the natural extension of the classical {\it Restricted Eigenvalue condition} introduced by \cite{BRT} to study the Lasso estimate. RE-type assumptions are among the mildest ones to establish oracle inequalities (see \cite{vdGB}). In the Gaussian setting, \cite{LPvdGT} considered similar conditions to establish oracle inequalities for their group-Lasso procedure.  In particular, if $c_0$ is a positive lower bound for $f_0$, then for all $\bbeta\in \R^p$, $$\bbeta^T\bG\bbeta = (\bA\bbeta)^T\bC(\bA\bbeta) \geq c_0 \|\bA\bbeta\|_2^2 = c_0 \sum_{i=1}^n \Big(\sum_{j=1}^p \beta_j \p_j(X_i)\Big)^2 = c_0 \sum_{i=1}^n g_\bbeta^2(X_i),$$ with $g_{\bbeta}=\sum_{j=1}^p \beta_j \p_j.$ If $(\p_j)_{j\in \J}$ is orthonormal  on $[0,1]^{d}$ and if the design is regular, then the last term is the same order as $$ n \int g_{\bbeta}^2(x) dx= n \|\bbeta\|_2^2 \geq n \|\bbeta_J\|_2^2$$ for any subset $J\subset \{1,\ldots,K\}$. Under these assumptions, $\kappa_n^{-2}(s,r)=O(n^{-1})$.

Under Assumption 1, we consider the slightly modified group-Lasso estimate. Let $\alpha>1$ and let us set $$\widehat\bbeta^{gL} \in \underset{\bbeta \in \Gamma(\mu)}{\argmin} \Big\{-l(\bbeta) +\alpha \sum_{k=1}^K \lambda_k^g \|\bbeta_{G_k}\|_2\Big\},\quad \widehat f^{gL}(x)=\exp\Biggl(\sum_{j=1}^p\widehat\beta_j^{gL}\p_j(x)\Biggr)$$ for which we obtain the following fast oracle inequality. 
\begin{thm}\label{fast_io_gl}
Let $\e >0$ and $s$ a positive integer. Let Assumption 2 be satisfied with $s$ and $$r= {\max_k \lambda_k^g \over \min_k \lambda_k^g}\frac{\alpha+1+2\alpha/\e}{\alpha-1}.$$ 
Then there exists a constant $B(\e,\mu)$ depending on $\e$ and $\mu$ such that, on $\Omega_g$, 
\begin{equation}\label{ineg_gl}
K(f_0,\widehat f^{gL}) \leq (1+\e)\inf_{\substack{\bbeta \in \Gamma(\mu)\\|J(\bbeta)|\leq s}}\Bigg\{K(f_0,f_\bbeta)+B(\e,\mu)\frac{\alpha^2|J(\bbeta)|}{\kappa_n^2}\times\left(\max_{k\in\{1,\ldots,K\}} {\lambda_k^g}\right)^2 \Bigg\},
\end{equation}
where $\kappa_n$ stands for $\kappa_n(s,r)$, and $J(\bbeta)$ is the subset of $\{1,\ldots,K\}$ such that $\bbeta_{G_k} =\zero$ if and only if $k \notin J(\bbeta)$.
\end{thm}
Let us comment each term of the right-hand side of (\ref{ineg_gl}). The first term is an approximation term, which can vanish if $f_0$ can be decomposed on the dictionary. The second term is a variance term, according to the usual terminology, which is proportional to the size of $J(\bbeta)$. Its shape is classical in the high dimensional setting. See for instance Theorem 3.2 of \cite{LPvdGT} for the group-Lasso in linear models, or Theorem 6.1 of \cite{BRT} and Theorem 3 of \cite{blpv} for the Lasso. If the order of magnitude of $\lambda_k^g$ is $\sqrt{n\times\max(\log p; |G_k|)}$ (see Section \ref{sec:order}) and if $\kappa_n^{-2}=O(n^{-1})$, the order of magnitude of this variance term is not larger than $|J(\bbeta)|\times \max(\log p; |G_k|)$.  Finally, if $f_0$ can be well approximated (for the empirical Kullback-Leibler divergence) by a group-sparse combination of the functions of the dictionary, then the right hand side of (\ref{ineg_gl}) will take small values. So, the previous result justifies our group-Lasso procedure from the theoretical point of view. Note that \eqref{slow_oi_gLasso} and (\ref{ineg_gl}) also show the interest of considering weights as small as possible. 

\cite{BLG14} established rates of convergence under stronger assumptions, namely all coordinates of the analog of $\bA$ are bounded by a quantity $L$, where $L$ is viewed as a constant. Rates depend on $L$ in an exponential manner and would highly deteriorate if $L$ depended on $n$ and $p$. So, this assumption is not reasonable if we consider   dictionaries such as wavelets or histograms (see Section~\ref{sec:order}).

\subsection{Oracle inequalities for the Lasso estimate}\label{ssec_io_single}
For the sake of completeness, we provide oracle inequalities for the Lasso. Theorems \ref{slow_io_gl} and \ref{fast_io_gl} that deal with the group-Lasso  estimate can be adapted to the non-grouping strategy when we take groups of size $1$. Subsequent results are similar to those established by \cite{Lemler} who studied the Lasso estimate for the high-dimensional Aalen multiplicative intensity model. The block $\ell_1$-norm $\|\cdot\|_{1,2}$ becomes the usual $\ell_1$-norm and the group support $J(\bbeta)$ is simply the support of $\bbeta$. As previously, we only work on the probability set $\Omega$ defined by
\begin{equation}\label{def_omega_prime}
\Omega= \Big\{|\bA_{j}^T(\bY-\E[\bY])| \leq \lambda_j\quad\forall j\in \{1,\ldots,p\}\Big\}.
\end{equation}
Theorem \ref{calibration} asserts that $\P(\Omega) \geq 1-{3\over p^{\gamma-1}}$ that goes to 1 as soon as $\gamma>1$. We obtain a slow oracle inequality for $\widehat f^{L}$:
\begin{Cor}\label{cor-simple}
On $\Omega$, $$ K(f_0,\widehat f^{L}) \leq \inf_{\bbeta \in \R^p} \Big\{K(f_0, f_\bbeta) + 2 \sum_{j=1}^p \lambda_j|\beta_j|\Big\}.$$
\end{Cor}

Now, let us consider fast oracle inequalities. In this framework, Assumption 2 is replaced with the following:

\medskip

\noindent\textbf{Assumption 3}. For some integer $s \in \{1,\ldots,p\}$ and some constant $r$, the following condition holds: $$0 < \kappa_n(s,r) = \min_{\substack{J\subset \{1,\ldots,p\}\\|J|\leq s}} \min_{\substack{\bbeta\in\R^p-\{0\}\\\|\bbeta_{J^c}\|_{1}\leq r\|\bbeta_J\|_{1}}} \frac{(\bbeta^T\bG\bbeta)^{1/2}}{\|\bbeta_J\|_2},$$ where $\bG$ is the Gram matrix defined by $\bG = \bA^T\bC\bA$, where $\bC$ is the diagonal matrix with $C_{i,i}=f_0(X_i)$. 

\medskip

Under Assumption 1, we consider the slightly modified Lasso estimate. Let $\alpha>1$ and let us set $$\widehat\bbeta^{L} \in \underset{\bbeta \in \Gamma(\mu)}{\argmin} \Big\{-l(\bbeta) +\alpha \sum_{j=1}^p \lambda_j |\bbeta_j|\Big\},\quad \widehat f^{L}(x)=\exp\Biggl(\sum_{j=1}^p\widehat\beta_j^{L}\p_j(x)\Biggr)$$ for which we obtain the following fast oracle inequality. 
\begin{Cor}\label{cor-RE}
Let $\e >0$ and $s$ a positive integer. Let Assumption 3 be satisfied with $s$ and  $$r= {\max_j \lambda_j \over \min_j \lambda_j}\frac{\alpha+1+2\alpha/\e}{\alpha-1}.$$ Then there exists a constant $B(\e,\mu)$ depending on $\e$ and $\mu$ such that, on $\Omega$,  $$K(f_0,\widehat f^{L}) \leq (1+\e)\inf_{\substack{\bbeta \in \Gamma(\mu)\\|J(\bbeta)|\leq s}}\Bigg\{K(f_0,f_\bbeta)+B(\e,\mu)\frac{\alpha^2|J(\bbeta)|}{\kappa_n^2}(\max_{j\in\{1,\ldots,p\}} {\lambda_j}^2) \Bigg\},$$ where $\kappa_n$ stands for $\kappa_n(s,r)$, and $J(\bbeta)$ is the support of $\beta$.
\end{Cor}
This corollary is derived easily from Theorem~\ref{fast_io_gl} by considering all groups of size 1. Comparing Corollary~\ref{cor-RE} and Theorem~\ref{fast_io_gl}, we observe that the group-Lasso can improve the Lasso estimate when the function $f_0$ can be well approximated by a function $f_\bbeta$ so that the number of non-zero groups of $\bbeta$ is much smaller than the total number of non-zero coefficients. The simulation study of the next section illustrates this comparison from the numerical point of view.

\section{Simulation study}\label{sec_simulation}

\paragraph{Simulation settings.} We explore the empirical performance of the Lasso and the group Lasso strategies using simulations. We considered different forms for intensity functions by taking the standard functions of \cite{DJ94}: blocks, bumps, doppler, heavisine, to set $g_0$. The signal to noise ratio was increased by multiplying the intensity functions by a factor $\alpha$ taking values in $\{1, \hdots,7\}$, $\alpha=7$ corresponding to the most favorable configuration. Observations $Y_i$ were generated such that $Y_i|X_i \sim {\mathcal Poisson}(f_0(X_i))$, with $f_0=\alpha \exp(g_0)$, and $(X_1,\hdots,X_n)$ was set as the regular grid of length $n=2^{10}$. Each configuration was repeated 20 times. Our method was implemented using the \texttt{grpLasso} \texttt{R} package of \cite{mvdgb} to which we provide our concentration-based weights. The corresponding code is available at \texttt{http://pbil.univ-lyon1.fr/members/fpicard/software.html}.

\paragraph{The basis and the dictionary frameworks.} The dictionary we consider is built on the (periodized) Haar and Daubechies basis, and on the Fourier basis, in order to catch piece-wise constant trends, localized peaks and periodicities. Each orthonormal system has $n$ elements, which makes $p=n$ when systems are considered separately, and $p=2n$ or $3n$ depending on the considered dictionary. For wavelets, the dyadic structure of the decomposition allows us to group the coefficients scale-wise by forming groups of coefficients of size $2^q$. As for the Fourier basis, groups (also of size $2^q$) are formed by considering successive coefficients (while keeping their natural ordering). When grouping strategies are considered, we set all groups at the same size. 

\paragraph{Weights calibration in practice.} First for both the Lasso and the group Lasso, we estimate $V_j$ (resp $V_j^g$) by $\widehat V_j$ (resp $\widehat V_j^g$) instead of using $\widetilde V_j$ (resp $\widetilde V_j^g$). This simplification is easier to compute in practice, and does not have any impact on the performance of the procedures. Lasso weights only depend on hyperparameter $\gamma$ that we choose equal to $1.01$, following the arguments at the end of Section \ref{ssec_lambdaj}. As for the group Lasso weights (Theorem \ref{theo_1_plus}), the first term is replaced by $\sqrt{\sum_{j\in G_k} \widehat V_j}$, as it is governed by a quantity that tends to one when $p$ is large. The second term was calibrated by using different values of $\gamma$, and the best empirical performance were achieved so that the left- and right-hand terms of (\ref{def_lambda_k^g}) were approximatively equal. This resumes to group-Lasso weights of the form $2 \sqrt{\sum_{j\in G_k} \widehat V_j}$.

\paragraph{Competitors.} We compete our Lasso procedure (\texttt{Lasso.exact} in the sequel), with the Haar-Fisz transform (for Haar and Daubechies systems) applied to the same data followed by soft-thresholding. Here we mention that we did not perform cycle-spinning (that is often included in denoising procedures) in order to focus on the effects of thresholding only. We also implemented the half-fold cross-validation proposed by \cite{NasonCV} in the Poisson case to set the weights in the penalty, with the proper scaling ($2^{s/2}\lambda$, with $s$ the scale of the wavelet coefficients) as proposed by \cite{sardy}. Then we compare the performance of the group-Lasso with varying group sizes (2,4,8) to the Lasso, to assess the benefits or grouping wavelet coefficients.

\paragraph{Performance measurement.} For any estimate $\hat f$, reconstruction performance were measured using the (normalized) mean-squared error $MSE = \|\widehat{f}-f_0\|_2^2/\|f_0\|_2^2$ (Figure \ref{Fig:simulMSE}), and selection performance were measured by the standard indicators: accuracy on support recovery, sensitivity (proportion of true non-null coefficients among selected coefficients) and specificity of detection (proportion of true null coefficients among non-selected coefficients), based on the estimated support of $\widehat{\bbeta}$ and on the support of $\bbeta_0$, the coefficients associated with the projection of function $f_0$ on the dictionary. 

\paragraph{Performance in the basis setting.} The first step of our simulation study relies on wavelet basis (Haar or Daubechies) and not on a dictionary approach (considered in a second step) in order to compare our calibrated weights with other methods that rely on penalized strategy. It appears that, except for the \texttt{bumps} function, the Lasso with exact weights shows the lowest reconstruction error whatever the shape of the intensity function (Figure \ref{Fig:simulMSE}). Moreover, better performance of the Lasso with exact weights in cases of low intensity emphasize the interest of theoretically calibrated procedures rather than asymptotic approximations (like the Haar-Fisz transform). In the case of \texttt{bumps}, cross-validation seems to perform better than the Lasso, but when looking at reconstructed average function (Figure \ref{Fig:simulreconsLasso}) this lower reconstruction error of cross-validation is associated with higher local variations around the peaks. Compared with Haar-Fisz, the gain of using exact weights is substantial even when the signal to noise ratio is high, which indicates that even in the validity domain of the Haar-Fisz transform (large intensities), the Lasso combined with exact thresholds is more suitable (Figure \ref{Fig:simulreconsLasso}). As for the group Lasso, its performance highly depend on the group size: while groups of size 2 show similar performance as the Lasso, groups of size 4 and 8 increase the reconstruction error (Figure \ref{Fig:simulMSE} and \ref{Fig:simulreconsgroup}), since they are not scaled to the size of the irregularities in the signal. This trend is not systematic as the group Lasso appears to be adapted to functions that are more regular (Heavisine), and seems to avoid edge effects in some situations. Very interestingly, the group Lasso of size 2 increases the sensitivity of detection for the Lasso (Figure \ref{Fig:simulSens}), while keeping the same specificity, which suggests that it accounts for (true) local variations of nearby coefficients, which results in a slightly better reconstruction error. As a last remark we mention that the sensitivities of all methods are rather low regarding coefficients selection, meaning that many true non null coefficients remain unselected. Since reconstruction errors are satisfactory, this means that only few coefficients needed to be selected for good reconstruction properties in the functional domain.

\paragraph{Performance in the dictionary framework.} Lastly, we explored the performance of the dictionary approach, by considering different dictionaries to estimate each function: Daubechies (D), Fourier (F), Haar (H), or their combinations (Figure \ref{Fig:simu3}). Rich dictionaries can be very powerful to catch complex shapes in the true intensity function (like the notch in the heavisine case Figure \ref{Fig:dicofun}), and the richest dictionary (DFH) often leads to the lowest reconstruction error (MSE) on average. However the richest dictionary (DFH) is not always the best choice in terms of reconstruction error, which is stricking in the case of the \texttt{blocks} function. In this case the Haar system only would be preferable for the Lasso (Figure \ref{Fig:dicoMSE}). For the group-Lasso and the \texttt{blocks} intensity function, the combination of the Daubechies and the Haar systems provides the best MSE, but when looking at the reconstructed intensity (Figure \ref{Fig:dicofun}-blocks), the Daubechies system introduces wiggles are not relevant for \texttt{blocks}. Also, richer dictionaries do not necessarily lead to more selected parameters (Figure \ref{Fig:dicoMSE}), which illustrates that selection depends on the redundancies between the systems elements of the dictionary. In practice we often do not have any \textit{prior} knowledge concerning the elements that shape the signal, and these simulations suggest that the blind use of the richest dictionary may not be the best strategy in terms of reconstructed functions. Consequently, in the following application, we propose to adapt the half-fold cross validation of \cite{NasonCV} to choose the best combinations of systems.

\begin{figure}
\centering
\subfloat[Mean Square error of reconstruction.]{
\includegraphics[scale=0.45]{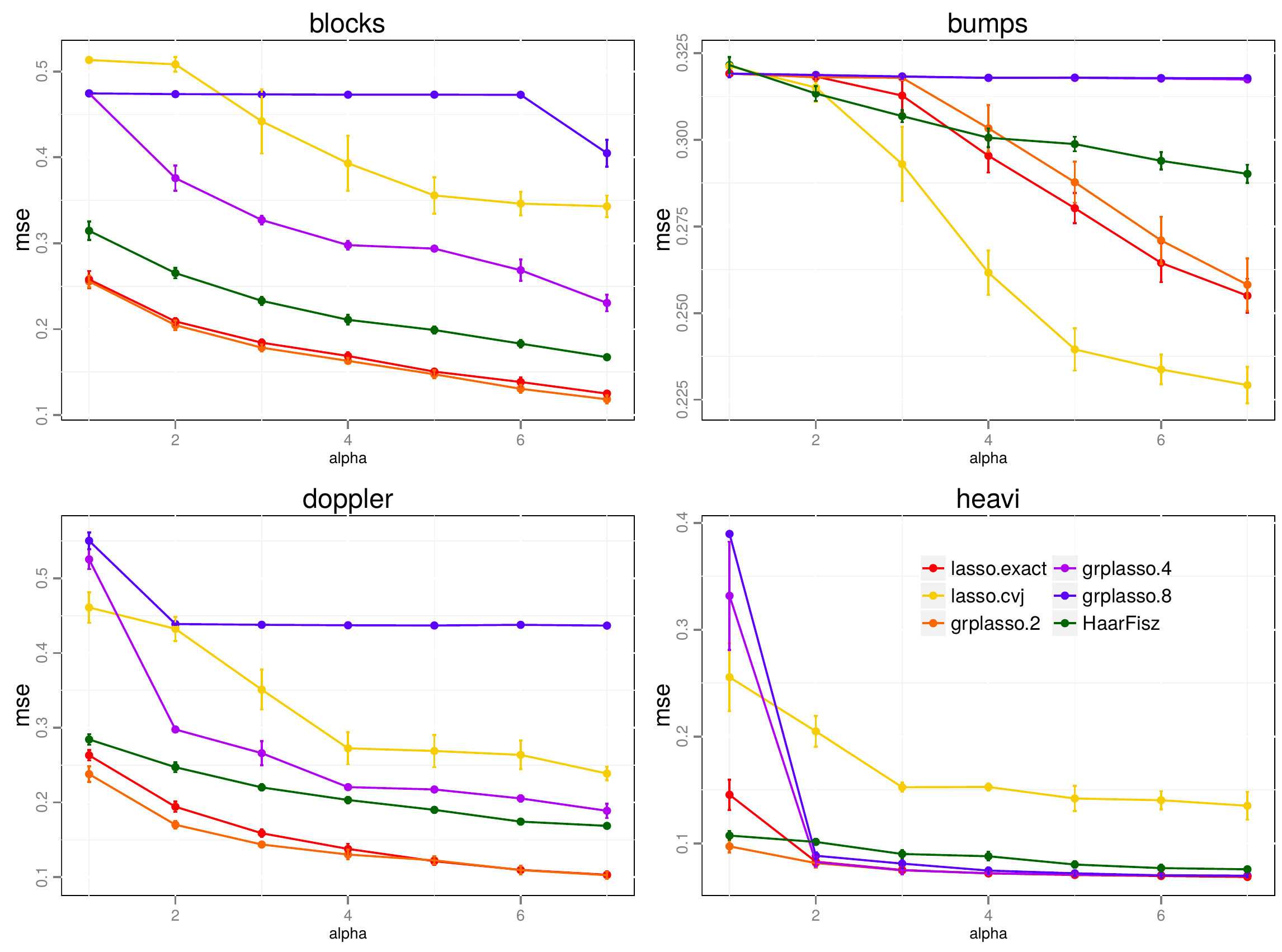}
\hfill
\label{Fig:simulMSE}}\\
\subfloat[Sensitivity of selection.]{
\includegraphics[scale=0.45,angle=90]{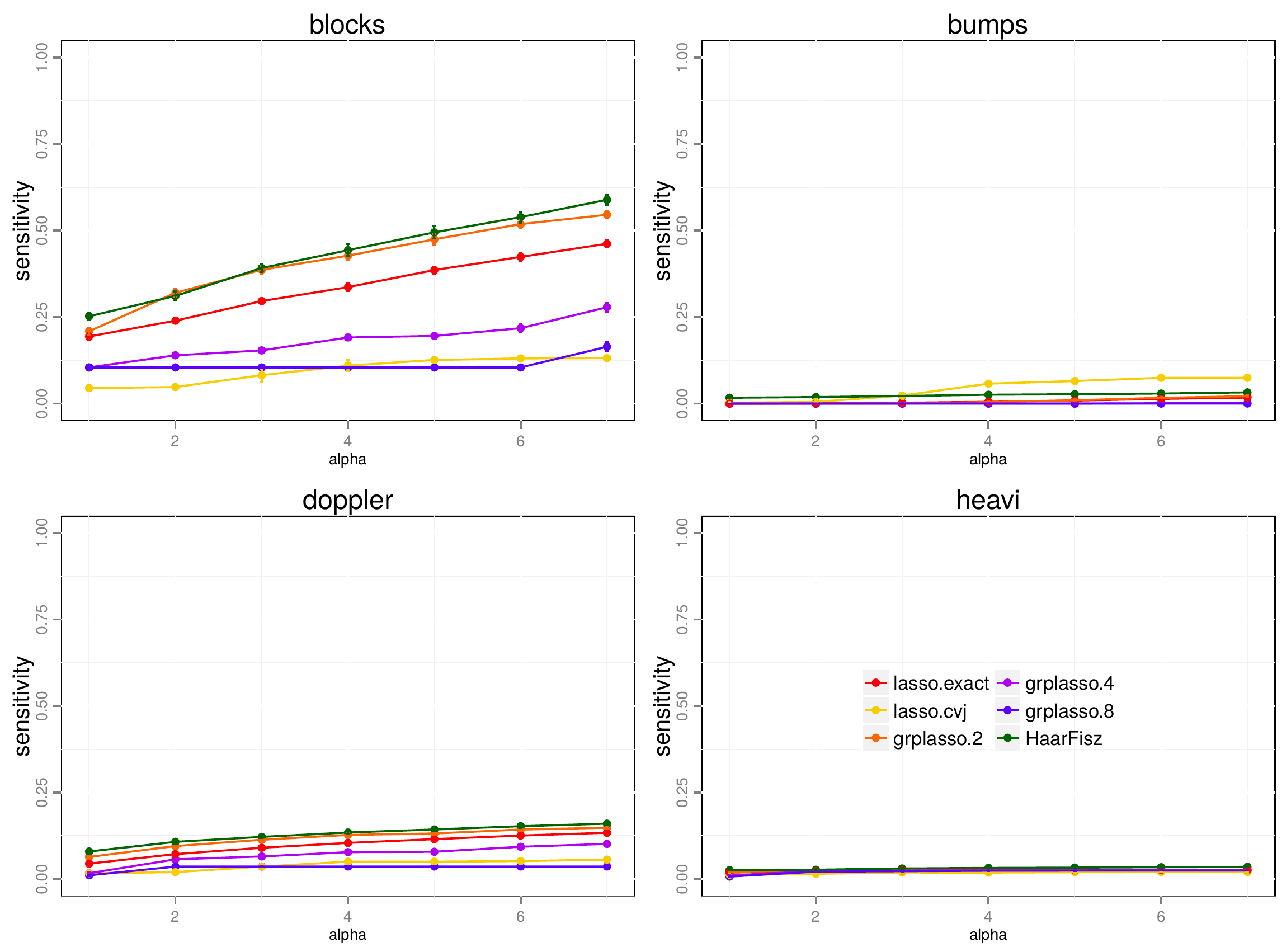}  
\hfill
\label{Fig:simulSens}}
\caption{Average (over 20 repetitions) Mean Square Error of reconstruction (\ref{Fig:simulMSE}) and sensitivity of selection (\ref{Fig:simulSens}) of different methods for the estimation of simulated intensity functions according to function shapes (blocks, bumps, doppler, heavisine) and signal strength ($\alpha$). \texttt{Lasso.exact}: Lasso penalty with our data-driven theoretical weights, \texttt{Lasso.cvj}: Lasso penalty with weights calibrated by cross validation with scaling $2^{s/2}\lambda$, \texttt{group.Lasso.2/4/8}: group Lasso penalty with our data-driven theoretical weights with group sizes 2/4/8, \texttt{HaarFisz}: Haar-Fisz tranform followed by soft-thresholding.
\label{Fig:simu1}}
\end{figure}

\begin{figure}
\centering
\subfloat[Average reconstructed functions for the Lasso and competitors.]{
\includegraphics[scale=0.5]{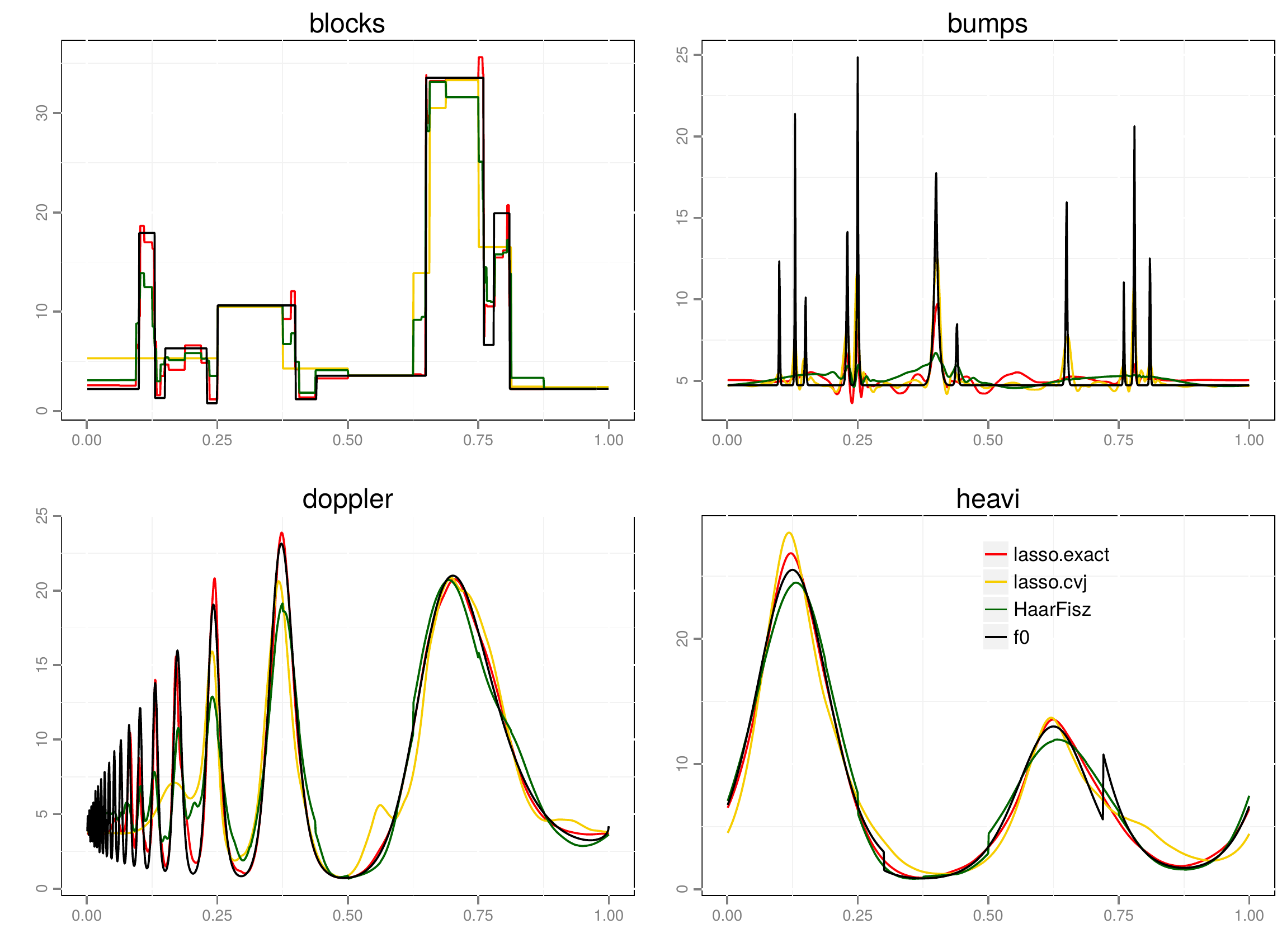} 
\hfill
\label{Fig:simulreconsLasso}}\\
\subfloat[Average reconstructed functions for the group strategies.]{
\includegraphics[scale=0.5]{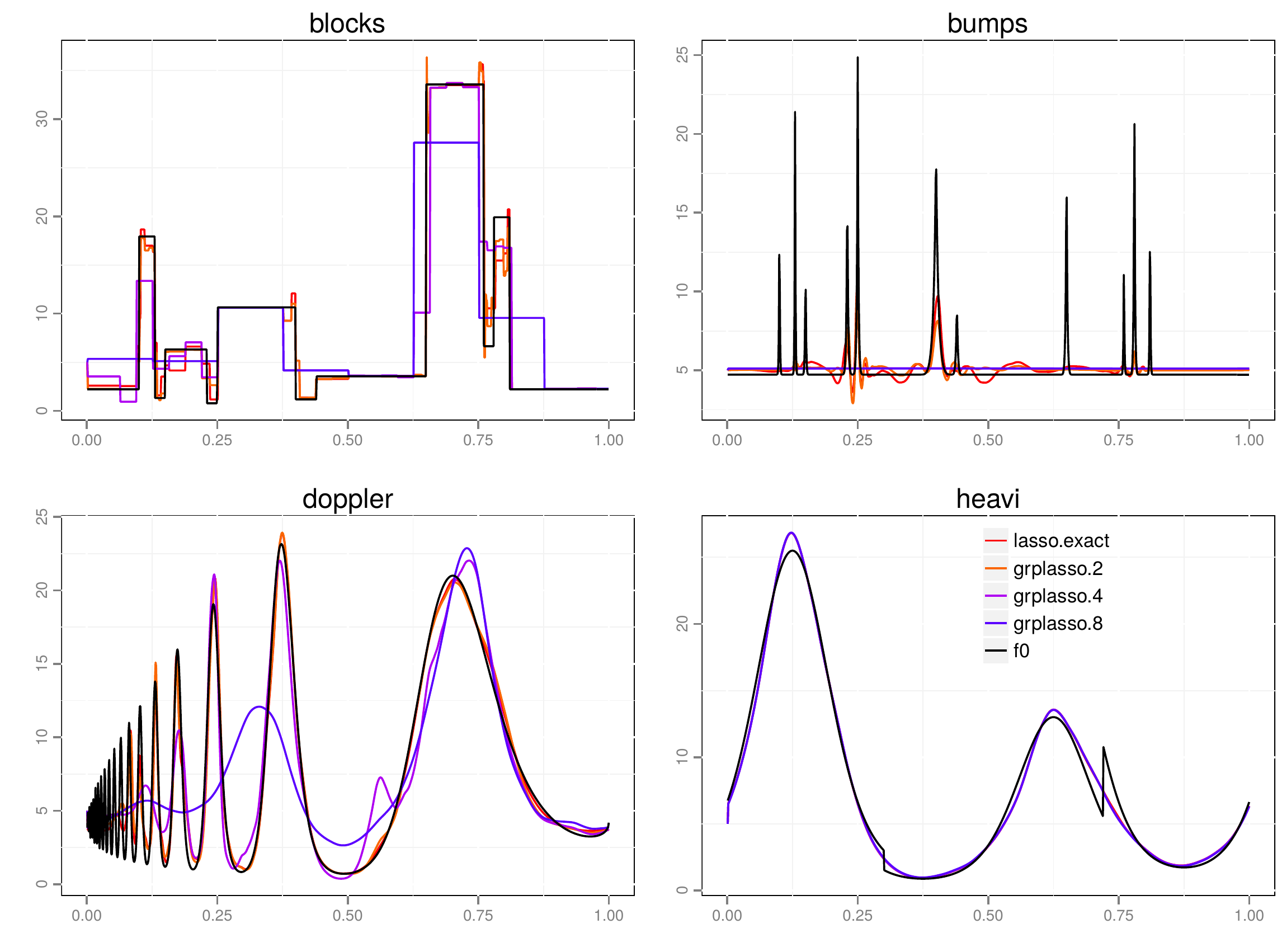}  
\hfill
\label{Fig:simulreconsgroup}}
\caption{Average (over 20 repetitions) reconstructed functions by different methods of estimation according to function shapes (blocks, bumps, doppler, heavisine). Top panel corresponds to non-grouped strategies (\ref{Fig:simulreconsLasso}) and bottom panel compares group-strategies to the Lasso (\ref{Fig:simulreconsgroup}). \texttt{Lasso.exact}: Lasso penalty with our data-driven theoretical weights, \texttt{Lasso.cvj}: Lasso penalty with weights calibrated by cross validation with scaling $2^{s/2}\lambda$, \texttt{group.Lasso.2/4/8}: group Lasso penalty with our data-driven theoretical weights with group sizes 2/4/8, \texttt{HaarFisz}: Haar-Fisz tranform followed by soft-thresholding, \texttt{f0}: simulated intensity function.
\label{Fig:simu2}}
\end{figure}

\begin{figure}
\centering
\subfloat[Average Mean Square Error for different dictionaries with respect to the average number of selected coefficients (df).]{
\includegraphics[scale=0.42]{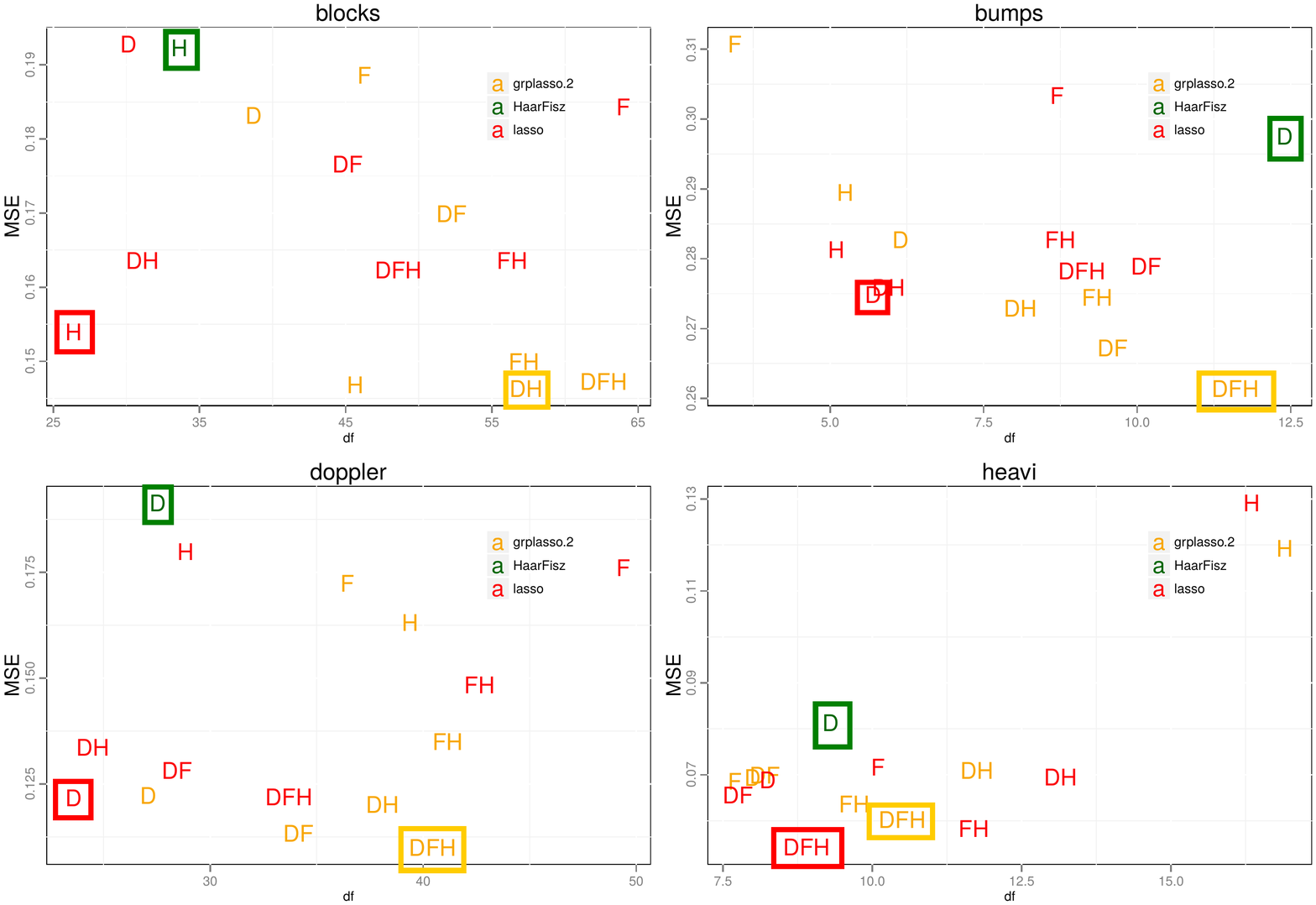} 
\hfill
\label{Fig:dicoMSE}}\\
\subfloat[Reconstructed functions for the dictionaries with the smallest MSE.]{
\includegraphics[scale=0.42]{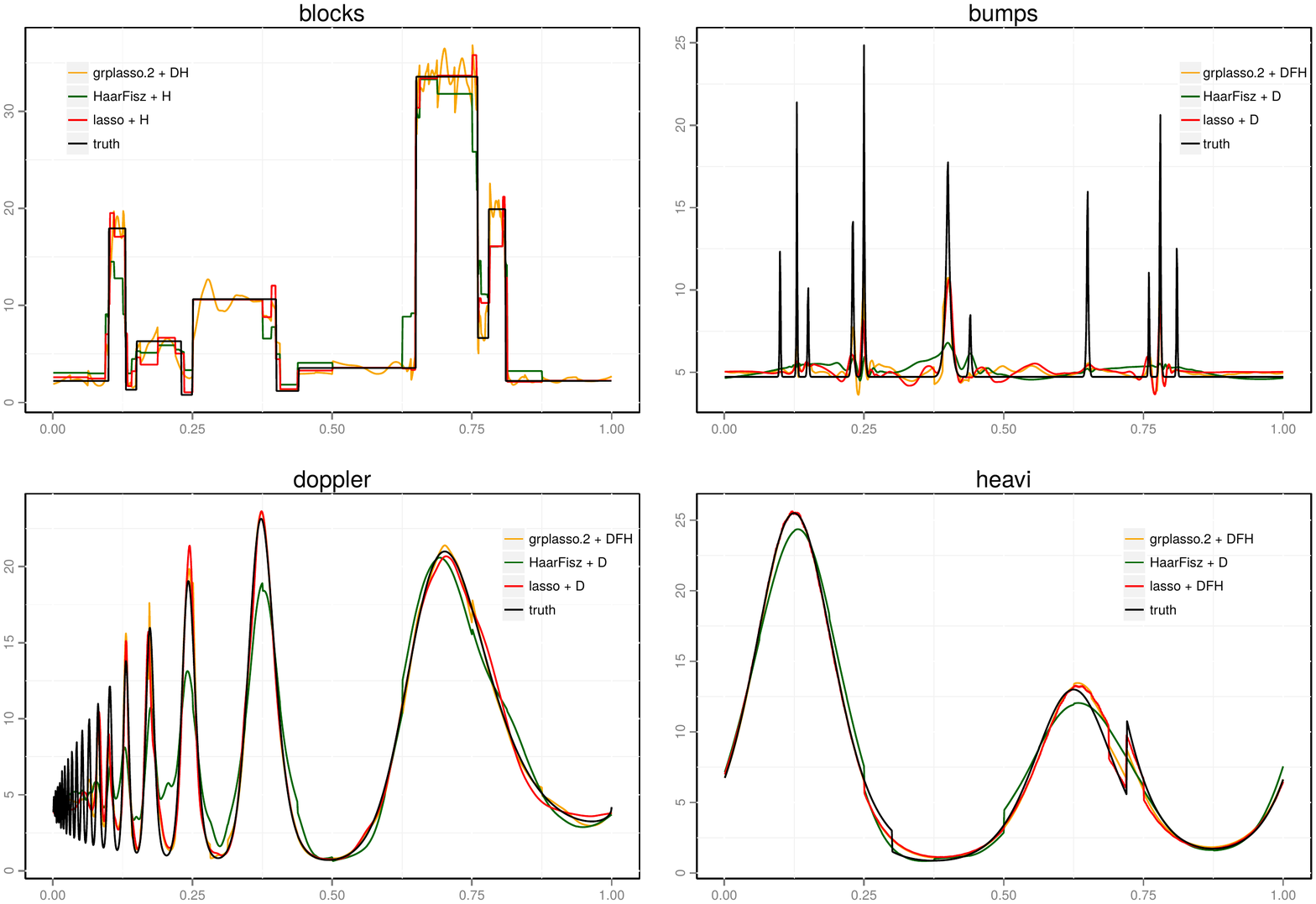}
\hfill
\label{Fig:dicofun}}
\caption{Average (over 20 repetitions) Mean Square Errors and number of selected coefficients (df) (\ref{Fig:dicoMSE}), and reconstructed functions (\ref{Fig:dicofun}) for different dictionaries: Daubechies (D), Fourier (F), Haar (H) and their combinations. \texttt{Lasso.exact}: Lasso penalty with our data-driven theoretical weights, \texttt{group.Lasso.2}: group Lasso penalty with our data-driven theoretical weights with group sizes 2, \texttt{HaarFisz}: Haar-Fisz tranform followed by soft-thresholding.
\label{Fig:simu3}}
\end{figure}

\section{Applications}\label{sec_applications}
The analysis of biological data has faced a new challenge with the extensive use of next generation sequencing (NGS) technologies. NGS experiments are based on the massive parallel sequencing of short sequences (reads). The mapping of these reads onto a reference genome (when available) generates counts data ($Y_t$) spatially organized (in 1D) along the genome (at position $X_t$). These technologies have revolutionized the perspectives of many fields in molecular biology, and among many applications, one is to get a local quantification of DNA or of a given DNA-related molecule (like transcription factors for instance with chIP-Seq experiments, \cite{F12}). This technology has recently been applied to the identification of replication origins along the human genome. Replication is the process by which a genome is duplicated into two copies. This process is tightly regulated in time and space so that the duplication process takes place in the highly regulated cell cycle. The human genome is replicated at many different starting points called origins of replication, that are loci along the genome at which the replication starts. Until very recently, the number of such origins remained controversial, and thanks to the application of NGS technologies, first estimates of this number could be obtained. The signal is made of counts along the human genome such that reads accumulations indicate an origin activity (see \cite{PCA14}). Scan statistics were first applied to these data, to detect significant local enrichments reads accumulation, but there is currently no consensus on the best method to analyze such data. Here we propose to use the Poisson functional regression to estimate the intensity function of the data on a portion of the human chromosomes X and 20. Half-fold cross-validation was used to select the appropriate dictionary between Daubechies, Fourier, Haar (and their combinations), and our theoretical weights were used to calibrate the Lasso (Figure \ref{Fig:Ori}). Our results are very promising as the sparse dictionary approach is very efficient for denoising (Chromosome X, Figure \ref{Fig:chrX}) and produces null intensities when the signal is low (higher specificity). Another aspect of our method is that it seems to be more powerful in the identification of peaks that are more precise (Chromosome 20, positions 0.20 and 0.25Mb, Figure \ref{Fig:chr20}), which indicates that the dictionary approach may be more sensitive to detect peaks. Given the spread of NGS data and the importance of peak detection in the analysis process, for chIP-Seq \cite{F12}, FAIRE-Seq \cite{TRH12}, OriSeq \cite{PCA14}, our preliminary results suggest that the sparse dictionary approach will be a very promising framework for the analysis of such data.

\begin{figure}
\centering
\subfloat[Chromosome 20]{
\includegraphics[scale=0.35]{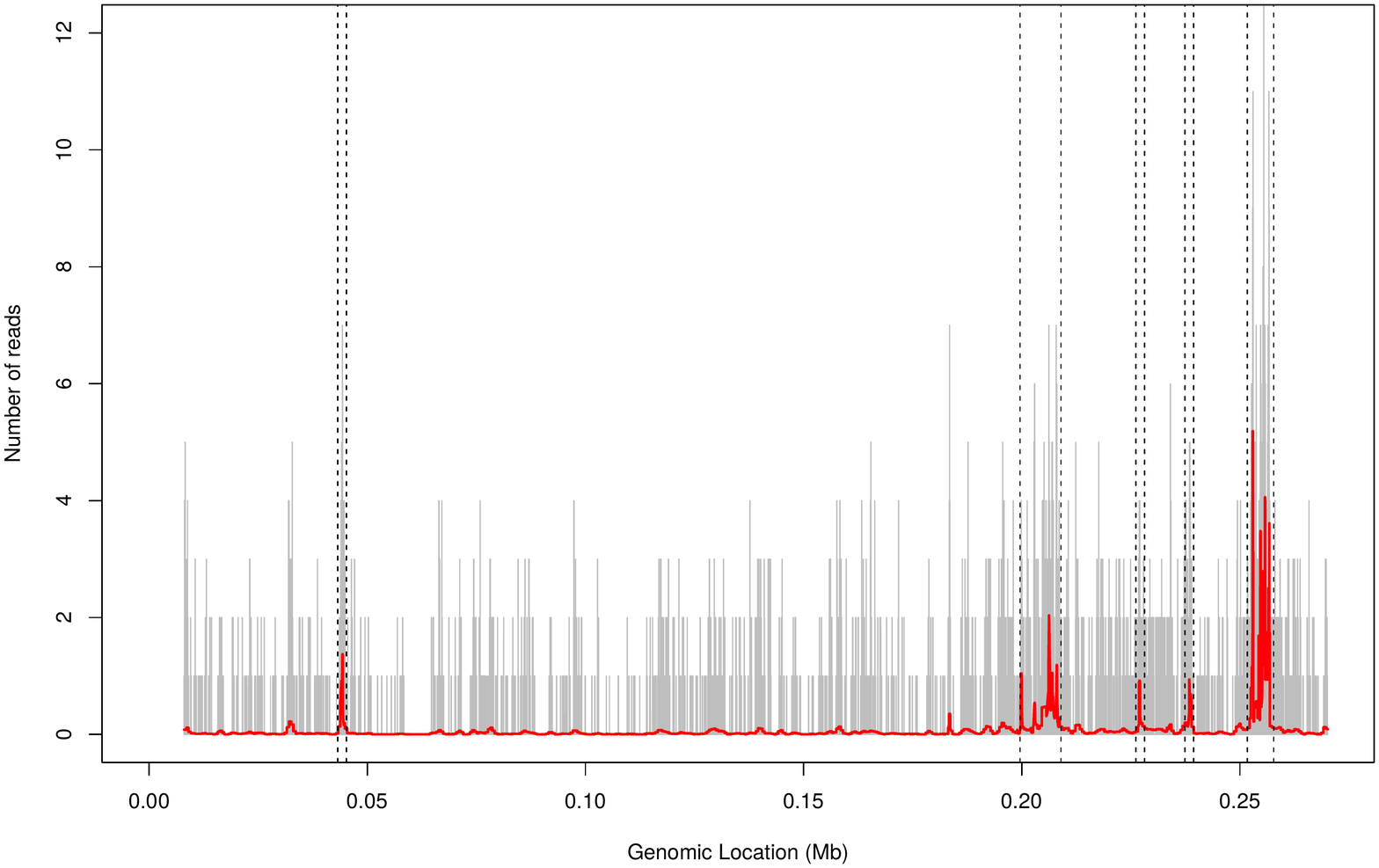}
\hfill
\label{Fig:chr20}}\\
\subfloat[Chromosome X]{
\includegraphics[scale=0.35]{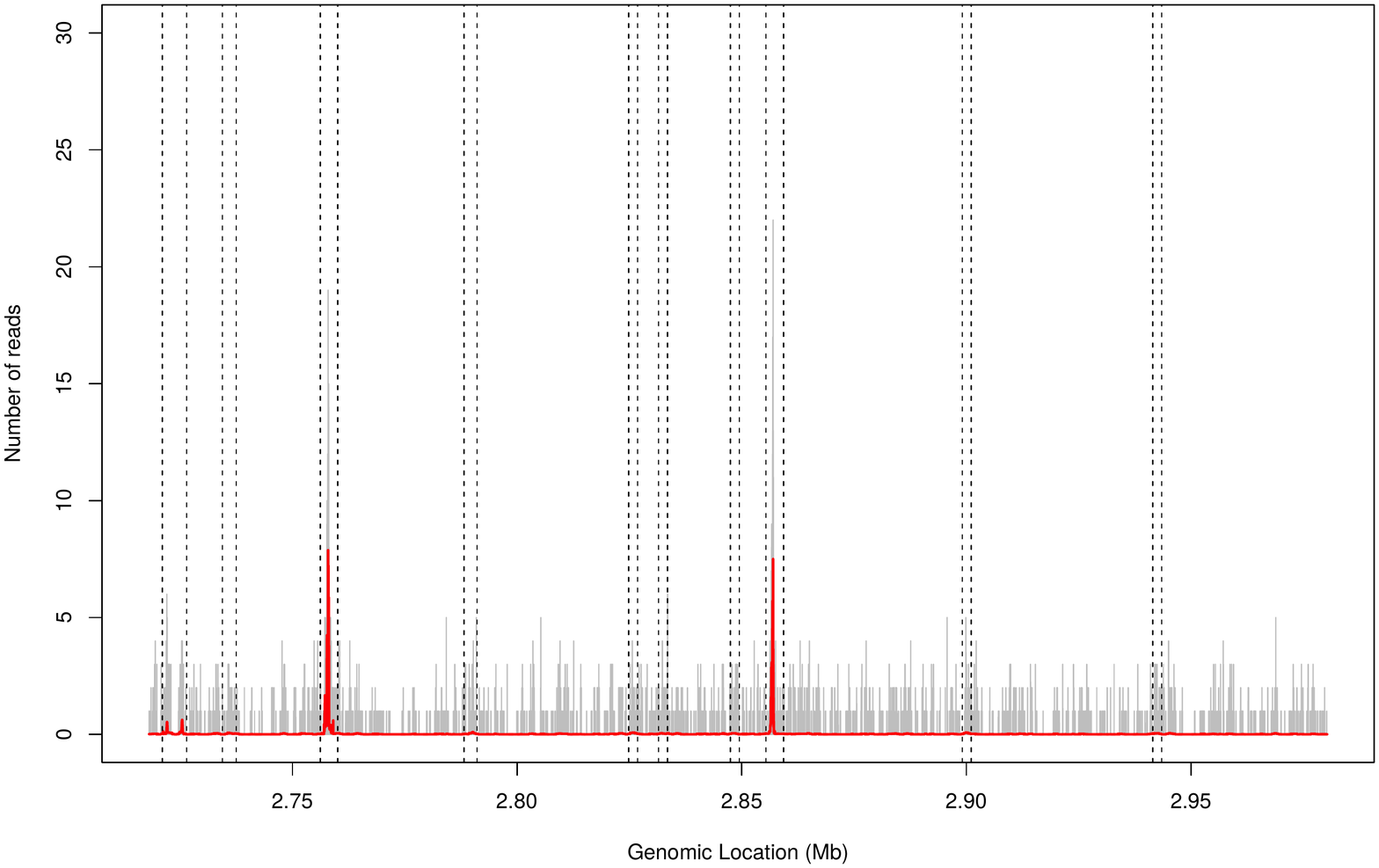}
\hfill
\label{Fig:chrX}}
\caption{Estimation of the intensity function of Ori-Seq data (chromosomes 20 \ref{Fig:chr20} and  X \ref{Fig:chrX}). Grey bars indicate the number of reads that match genomic positions (x-axis, in MegaBases). The red line corresponds to the estimated intensity function, and vertical dotted lines stand for the detected origins by scanning statistics.\label{Fig:Ori}}
\end{figure}

\newpage
\section{Proofs}\label{sec_proofs}
\subsection{Proof of Theorem \ref{calibration}}\label{sec:proof1}
We denote  by $\mu$ the Lebesgue measure on $\R^d$ and we introduce a partition of the set $[0,1]^d$ denoted $\cup_{i=1}^nS_i$ so that for any $i=1,\ldots,n$, $X_i\in S_i$ and $\mu(S_i)>0$. Let $h$ the function defined for any $t\in[0,1]^d$ by
$$h(t)=\sum_{i=1}^n\frac{f_0(X_i)}{\mu(S_i)}1_{S_i}(t).$$
Finally, we introduce $N$ the Poisson process on $[0,1]^d$ with intensity $h$ (see \cite{Kingman}). Therefore, for any $i=1,\ldots,n$, $N(S_i)$ is a Poisson variable with parameter $\int _{S_i} h(t)dt=f_0(X_i)$ and since $\cup_{i=1}^nS_i$ is a partition of $[0,1]^d$, $(N(S_1),\ldots,N(S_n))$ has the same distribution as $(Y_1,\ldots,Y_n)$. We observe that if for any $j=1,\ldots,p$, 
$$\widetilde \p_j(t)=\sum_{i=1}^n \p_j(X_i) 1_{S_i}(t),$$
then
$$\int \widetilde \p_j(t)dN(t)\sim\sum_{i=1}^n \p_j(X_i)Y_i=\bA_j^T\bY.$$
We use the following exponential inequality (see Inequality (5.2) of \cite{ptrf}). If $g$ is bounded, for any $u>0$,
  \begin{equation}
  \label{invformben}
 \P\left(\int g(x) (dN(x)-h(x)dx) \geq
    \sqrt{2u\int g^2(x) h(x)dx}+\frac{u}{3}\norm{g}_\infty\right)\leq\exp(-u).
 \end{equation}
By taking successively $g=\widetilde\p_j$ and $g=-\widetilde\p_j$, we obtain
$$\P\Bigg(|\bA_j^T(\bY-\E[\bY])|\geq \sqrt{2u\int\widetilde\p_j^2(x)h(x)dx}+\frac{u}{3}\|\widetilde\p_j\|_\infty\Bigg)\leq 2e^{-u}.$$
Since
$$\int\widetilde\p_j^2(x)h(x)dx = \sum_{i=1}^n \p_j^2(X_i) f_0(X_i)= V_j,
$$
we obtain
\begin{equation}\label{lemme_part_1}
\P\Bigg(|\bA_j^T(\bY-\E[\bY])|\geq \sqrt{2uV_j}+\frac{u}{3}\|\widetilde\p_j\|_\infty\Bigg)\leq 2e^{-u}.
\end{equation}
To control $V_j$, we use \eqref{invformben} with $g=-\widetilde\p_j^2$ and we have:
$$\P\Bigg(V_j - \widehat V_j \geq \sqrt{2u\int\widetilde\p_j^4(t)h(t)dt}+\frac{u}{3}\|\widetilde\p_j\|_{\infty}^2\Bigg)\leq e^{-u}.$$
We observe that 
$$\int\widetilde\p_j^4(t)h(t)dt \leq \|\widetilde\p_j\|_{\infty}^2\int\widetilde\p_j^2(t)h(t)dt=\|\widetilde\p_j\|_{\infty}^2 V_j.$$
Setting $v_j = u\|\widetilde\p_j\|_{\infty}^2,$ we have:
$$\P\Bigg(V_j-\sqrt{2v_jV_j}-\frac{v_j}{3} -\widehat V_j \geq 0\Bigg)\leq e^{-u}.$$
Let $\alpha_j = \sqrt{\widehat V_j + \frac{5}{6}v_j}+\sqrt{\frac{v_j}{2}},$
such that $\alpha_j$ is the positive solution to $\alpha_j^2-\sqrt{2v_j}\alpha_j-(\widehat V_j+\frac{v_j}{3})=0.$
Then
\begin{equation}\label{lemme_part_2}
\P\Big(V_j\geq \alpha_j^2\Big) = \P\Big(\sqrt{V_j}\geq\alpha_j\Big)\leq e^{-u}.
\end{equation}
We choose $u = \gamma \log p$ and observe that $\alpha_j^2 \leq \widetilde V_j.$ Then, by combining (\ref{lemme_part_1}) and (\ref{lemme_part_2}), we have
$$\P\Bigg(|\bA_j^T(\bY-\E[\bY])|\geq\sqrt{2\gamma\log p\widetilde V_j}+\frac{\gamma \log p}{3}\|\widetilde\p_j\|_{\infty}\Bigg)\leq {3\over p^\gamma}.$$
As $\|\widetilde\p_j\|_{\infty}= \max_i |\p_j(X_i)|$, the theorem follows. \hfill$\square$
\begin{Rem}\label{order}
By slightly extending previous computations, we easily show that for $u>0$,
$$\P\Bigg(|V_j - \widehat V_j| \geq \sqrt{2uV_j\|\widetilde\p_j\|_{\infty}^2}+\frac{u}{3}\|\widetilde\p_j\|_{\infty}^2\Bigg)\leq 2e^{-u},$$
which leads to 
$$\P\Bigg(|V_j - \widehat V_j| \geq \frac{V_j}{2}+\frac{4\gamma\log p}{3}\|\widetilde\p_j\|_{\infty}^2\Bigg)\leq \frac{2}{p^\gamma}.$$
\end{Rem}
\subsection{Proof of Theorem \ref{theo_1_plus}}
For each $k \in \{1,\ldots,K\}$, we recall that $b_k^i = \sqrt{\sum_{j \in G_k}\p_j^2(X_i)}$, so $b_k^i = \|\bA_{G_k}^T\be_i\|_2$, where $\be_i$ is the vector whose $i$-th coordinate is equal to 1 and all others to 0. We first state the following lemma:
\begin{lem}\label{theo_1}
 Let $k$ be fixed. Assume that there exists some $M>0$ such that $\forall\, x, |f_0(x)|\leq M$.\\
Assume further that there exists some $c_k \geq 0$ such that $\forall\, \by \in \R^n, \|\bA_{G_k}\bA_{G_k}^T\by\|_2 \leq c_k \|\bA_{G_k}^T\by\|_2$.\\
 Then, $\forall\, x>0, \forall\, \e>0$,
 \begin{equation}\label{theoreme}
  \P\Bigg(\|\bA_{G_k}^T(\bY-\E[\bY])\|_2 \geq (1+\e)\sqrt{\sum_{j\in G_k} V_j} +x \Bigg) \leq \exp\Bigg(\frac{x}{b_k} - \Big(\frac{x}{b_k} + \frac{D_k^\e}{b_k^2}\Big)\log\Big(1 + \frac{b_kx}{D_k^\e}\Big)\Bigg),\nonumber
 \end{equation}
 where $D_k^\e = 8Mc_k^2 + \frac{2}{\e^2}b_k^2$.
\end{lem}
\begin{proof}
 With $k \in \{1,\ldots,K\}$ being fixed, we define $f : \R^n \to \R$ by $f(\by) = \Big(\|\bA_{G_k}^T\by\|_2 - E\Big)_+$, where $E>0$ is a constant chosen later. We use Corollary 1 from~\cite{HMRB}, applied to the infinitely divisible vector $\bY-\E[\bY] \in \R^n$, whose components are independent, and to $f$. First note that for any $t>0$,
\begin{eqnarray*}
 \E e^{tb_k^i|Y_i-\E Y_i|} &\leq& \E e^{tb_k^i(Y_i+f_0(X_i))}\\
&=& \exp\Big(f_0(X_i)(e^{tb_k^i}+tb_k^i-1)\Big)<\infty.
\end{eqnarray*}
Furthermore, for any $i \in \{1,...,n\}$, any $\by\in\R^n$ and any $u \in \R$,
\begin{eqnarray*}
 |f(\by+u\be_i) - f(\by)| &\leq& \Big|\|\bA_{G_k}^T(\by+u\be_i)\|_2 - \|\bA_{G_k}^T\by\|_2\Big|\\
&\leq& \|\bA_{G_k}^T(u\be_i)\|_2\\
&=& |u|b_k^i.
\end{eqnarray*}
Therefore, for all $x>0,$
$$\P\Big(f(\bY-\E[\bY]) - \E[f(\bY-\E[\bY])] \geq x\Big) \leq \exp\Big(-\int_0^x h_f^{-1}(s)ds\Big),$$
where $h_f$ is defined for all $t>0$ by
$$h_f(t) = \sup_{\by \in \R^n} \sum_{i=1}^n \int_\R |f(\by+u\be_i) - f(\by)|^2 \frac{e^{tb_k^i|u|}-1}{b_k^i|u|}\widetilde\nu_i(du)$$
and $\widetilde\nu_i$ is the L\'evy measure associated with $Y_i-\E[Y_i]$.
It is easy to show that $\widetilde\nu_i = f_0(X_i)\delta_1$, and so
$$h_f(t) = \sup_{\by\in \R^n} \sum_{i=1}^n f_0(X_i) \Big(f(\by+\be_i) - f(\by)\Big)^2 \frac{e^{tb_k^i}-1}{b_k^i}.$$
 Furthermore, writing $A_i = \Big\{\|\bA_{G_k}^T(\by+\be_i)\|_2 \geq E \text{ or } \|\bA_{G_k}^T\by\|_2 \geq E\Big\}$, we have
\begin{eqnarray*}
 |f(\by+\be_i) - f(\by)| &\leq& \Big|\|\bA_{G_k}^T(\by+\be_i)\|_2 - \|\bA_{G_k}^T\by\|_2\Big| 1_{A_i}\\
&=& \frac{1_{A_i}\Big|\|\bA_{G_k}^T(\by+\be_i)\|_2^2 - \|\bA_{G_k}^T\by\|_2^2\Big|}{\|\bA_{G_k}^T(\by+\be_i)\|_2 + \|\bA_{G_k}^T\by\|_2}\\
&=& \frac{1_{A_i}\Big|2 <\bA_{G_k}^T\be_i,\bA_{G_k}^T\by> + \|\bA_{G_k}^T\be_i\|_2^2\Big|}{\|\bA_{G_k}^T(\by+\be_i)\|_2 + \|\bA_{G_k}^T\by\|_2}\\
&\leq& 2\frac{\Big|<\bA_{G_k}^T\be_i,\bA_{G_k}^T\by>\Big|}{\|\bA_{G_k}^T\by\|_2}+\frac{\|\bA_{G_k}^T\be_i\|_2^2}{E},
\end{eqnarray*}
with $<\cdot,\cdot>$ the usual scalar product. We now have
$$\Big(f(\by+\be_i)-f(\by)\Big)^2 \leq 8\frac{<\bA_{G_k}^T\be_i,\bA_{G_k}^T\by>^2}{\|\bA_{G_k}^T\by\|_2^2} + 2\frac{\|\bA_{G_k}^T\be_i\|_2^4}{E^2}.$$
The first term can be rewritten as $8\frac{<\be_i,\bA_{G_k}\bA_{G_k}^T\by>^2}{\|\bA_{G_k}^T\by\|_2^2}$ and the second one is equal to $2\frac{{b_k^i}^4}{E^2}$, so we can now bound $h_f(t)$ as follows.
\begin{eqnarray*}
 h_f(t) &\leq& \sup_\by \sum_i f_0(X_i)\frac{e^{tb_k^i}-1}{b_k^i}\Bigg(8\frac{<\be_i,\bA_{G_k}\bA_{G_k}^T\by>^2}{\|\bA_{G_k}^T\by\|_2^2} + 2\frac{{b_k^i}^4}{E^2}\Bigg)\\
&\leq& \frac{e^{tb_k}-1}{b_k} \sup_\by \Bigg(8M\frac{\|\bA_{G_k}\bA_{G_k}^T\by\|_2^2}{\|\bA_{G_k}^T\by\|_2^2} + \frac{2}{E^2}\sum_i f_0(X_i){b_k^i}^4\Bigg)\\
&\leq& \frac{e^{tb_k}-1}{b_k} \Bigg(8Mc_k^2 + \frac{2}{E^2}\sum_i f_0(X_i){b_k^i}^4\Bigg).
\end{eqnarray*}
Now, we set 
$$E = \e\sqrt{\sum_{j \in G_k} V_j}.$$
So we have:
\begin{eqnarray*}
E^2 
&=& \e^2 \sum_{j\in G_k} \sum_{i=1}^n f_0(X_i) \p_j^2(X_i)\\
&=& \e^2 \sum_{i=1}^n f_0(X_i) \sum_{j\in G_k} \p_j^2(X_i)\\
&=& \e^2 \sum_{i=1}^n f_0(X_i) {b_k^i}^2.
\end{eqnarray*}
Thus, we can finally bound the function $h_f$ by the increasing function $h$ defined by
$$h(t) = D_k^\e\frac{e^{tb_k}-1}{b_k},$$
with $D_k^\e = 8Mc_k^2 + \frac{2b_k^2}{\e^2}$. Therefore,
\begin{eqnarray*}
 \exp\Big(-\int_0^x h_f^{-1}(s)ds\Big) &\leq& \exp\Big(-\int_0^x h^{-1}(s)ds\Big)\\
&=& \exp\Bigg(\frac{x}{b_k} - \Big(\frac{x}{b_k} + \frac{D_k^\e}{b_k^2}\Big)\log\Big(1+\frac{b_kx}{D_k^\e}\Big)\Bigg).
\end{eqnarray*}
Now,
\begin{eqnarray*}
 f(\bY-\E[\bY]) - \E[ f(\bY-\E[\bY])] &=& \Big(\|\bA_{G_k}^T(\bY-\E[\bY])\|_2 - E\Big)_+ - \E\Big(\|\bA_{G_k}^T(\bY-\E[\bY])\|_2 - E\Big)_+\\
&\geq& \|\bA_{G_k}^T(\bY-\E[\bY])\|_2 - E - \E\|\bA_{G_k}^T(\bY-\E[\bY])\|_2.
\end{eqnarray*}
Furthermore, by Jensen's inequality, we have 
\begin{eqnarray*}
 \E\|\bA_{G_k}^T(\bY-\E[\bY])\|_2 &\leq& \sqrt{\E\|\bA_{G_k}^T(\bY-\E[\bY])\|_2^2}\\
&=& \sqrt{\sum_{j \in G_k} \E[(\bA_j^T(\bY-\E \bY))^2]}\\
&=& \sqrt{\sum_{j \in G_k} \Var(\bA_j^T\bY)}\\
&=& \sqrt{\sum_{j \in G_k} V_j}.
\end{eqnarray*}
Recalling that $E = \e\sqrt{\sum_{j \in G_k} V_j}$, we thus have
$$\P\Big(f(\bY-\E[\bY]) - \E f(\bY-\E[\bY]) \geq x\Big) \geq \P\Big(\|\bA_{G_k}^T(\bY-\E[\bY])\|_2 - (1+\e)\sqrt{\sum_{j\in G_k}V_j}\geq x\Big),$$
which concludes the proof.
\end{proof}
We apply Lemma~\ref{theo_1} with 
$$\e=\frac{1}{2\sqrt{2\gamma\log p}}\quad\mbox{ and }\quad x=2 \sqrt{\gamma\log p D_k^\e}.$$
Then,
\begin{eqnarray*}
 \frac{b_kx}{D_k^\e}&=& \frac{2 b_k\sqrt{\gamma\log p}}{\sqrt{D_k^\e}}\\
 &=&\frac{2 b_k\sqrt{\gamma\log p}}{\sqrt{8Mc_k^2 + \frac{2b_k^2}{\e^2}}}\\
 &\leq&\e\sqrt{2\gamma\log p}=\frac{1}{2}.
\end{eqnarray*}
Finally, using the fact that $\log(1+u) \geq u - \frac{u^2}{2}$, we have:
\begin{eqnarray*}
 \exp\Bigg(\frac{x}{b_k} - \Big(\frac{x}{b_k} + \frac{D_k^\e}{b_k^2}\Big)\log\Big(1 + \frac{b_kx}{D_k^\e}\Big)\Bigg) &\leq& \exp\Bigg(\frac{x}{b_k} - \Big(\frac{x}{b_k} + \frac{D_k^\e}{b_k^2}\Big)\Big(\frac{b_k x}{D_k^\e} - \frac{b_k^2x^2}{2{D_k^\e}^2}\Big)\Bigg)\\
&=& \exp\Bigg(\frac{-x^2}{2D_k^\e}+\frac{b_kx^3}{2{D_k^\e}^2} \Bigg)\\
&=& \exp\Bigg(\frac{-x^2}{2D_k^\e}\Big(1-\frac{b_kx}{D_k^\e}\Big)\Bigg)\\
&\leq&\exp\Big(\frac{-x^2}{4D_k^\e}\Big) = \frac{1}{p^\gamma}.
\end{eqnarray*}
We obtain
$$ \P\Bigg(\|\bA_{G_k}^T(\bY-\E[\bY])\|_2 \geq (1+\e)\sqrt{\sum_{j\in G_k} V_j} + 2 \sqrt{\gamma\log p D_k^\e}\Bigg) \leq \frac{1}{p^\gamma}.$$
We control $V_j$ as in the proof of Theorem \ref{calibration}, but we take $u = \gamma\log p + \log |G_k|$. The analog of (\ref{lemme_part_2}) is
$$\P\Big(V_j > \widetilde V_j^g \Big) \leq e^{-u}= \frac{1}{|G_k|p^\gamma}$$
and thus
\begin{eqnarray*}
 \P\Big( \exists\, j \in G_k, \,V_j > \widetilde V_j^g\Big) &\leq& \frac{1}{p^\gamma}.
\end{eqnarray*}
This concludes the proof of Theorem \ref{theo_1_plus}. \hfill$\square$
\subsection{Proof of Proposition~\ref{prop}}
For the first point, we write:
$$  \|\bA_{G_k}\bA_{G_k}^T\bx\|_2^2 = \sum_{l=1}^n \Bigg(\sum_{j \in G_k} \p_j(X_l)\sum_{i=1}^n \p_j(X_i)x_i\Bigg)^2.$$
Then, we apply the Cauchy-Schwarz inequality:
\begin{eqnarray*}
 \|\bA_{G_k}\bA_{G_k}^T\bx\|_2^2 &\leq& \sum_{l=1}^n \Bigg(\sum_{j \in G_k} \p_j^2(X_l)\Bigg)\Bigg(\sum_{j \in G_k}\Big(\sum_{i=1}^n \p_j(X_i)x_i\Big)^2\Bigg)\\
&=& \|\bA_{G_k}^T\bx\|_2^2 \sum_{l=1}^n \Bigg(\sum_{j \in G_k} \p_j^2(X_l)\Bigg)\\
&=&\|\bA_{G_k}^T\bx\|_2^2 \sum_{l=1}^n(b_k^l)^2\\
&\leq& n b_k^2\|\bA_{G_k}^T\bx\|_2^2,
\end{eqnarray*}
which proves the upper bound of (\ref{c_inf_b}).  For the lower bound, we just observe that for any $i=1,\ldots,n$, with $\be_i$ the vector whose $i$-th coordinate is equal to 1 and all others to 0,
\begin{eqnarray*}
 {b_k^i}^2 &=&\|\bA_{G_k}^T\be_i\|_2^2\\
&=& <\bA_{G_k}^T\be_i, \bA_{G_k}^T\be_i>\\
&=& <\be_i, \bA_{G_k}\bA_{G_k}^T\be_i>\\
&\leq& \|\be_i\|_2 \|\bA_{G_k}\bA_{G_k}^T\be_i\|_2\\
&\leq& c_k \|\bA_{G_k}^T\be_i\|_2\\
&=& c_k b_k^i,
\end{eqnarray*}
which obviously entails $b_k\leq c_k$.
For the last point, we observe that $$\|\bA_{G_k}^T\bx\|_2^2 = \sum_{j\in G_k} K_j^2,$$
where $K_j = \sum_{i=1}^n \p_j(X_i)x_i$.
By expressing $\|\bA_{G_k}\bA_{G_k}^T\bx\|_2^2$ with respect to the $K_j$'s, we obtain:
\begin{eqnarray*}
 \|\bA_{G_k}\bA_{G_k}^T\bx\|_2^2 &=& \sum_{l=1}^n \Bigg(\sum_{j \in G_k} \p_j(X_l)\sum_{i=1}^n \p_j(X_i)x_i\Bigg)^2\\
&=& \sum_{l=1}^n \sum_{j\in G_k} \p_j(X_l)\sum_{i=1}^n\p_j(X_i)x_i \sum_{j'\in G_k} \p_{j'}(X_l)\sum_{i'=1}^n\p_{j'}(X_{i'})x_{i'}\\
&=& \sum_{j\in G_k}\sum_{j'\in G_k}\sum_{l=1}^n\p_j(X_l)\p_{j'}(X_l)\sum_{i=1}^n\p_j(X_i)x_i\sum_{i'=1}^n\p_{j'}(X_{i'})x_{i'}\\
&=& \sum_{j\in G_k}\sum_{j'\in G_k}\sum_{l=1}^n\p_j(X_l)\p_{j'}(X_l)K_j K_{j'}\\
&\leq& \frac{1}{2}\sum_{j\in G_k}\sum_{j'\in G_k}\Big|\sum_{l=1}^n\p_j(X_l)\p_{j'}(X_l)\Big|(K_j^2+ K_{j'}^2)\\
&=&\sum_{j\in G_k}\sum_{j'\in G_k}\Big|\sum_{l=1}^n\p_j(X_l)\p_{j'}(X_l)\Big|K_j^2,
\end{eqnarray*}
from which we deduce (\ref{majock}). \hfill$\square$
\subsection{Proof of Theorem \ref{slow_io_gl}}
For any $\bbeta\in\R^p$, we have
\begin{eqnarray}
K(f_0,f_\bbeta) &=&  \sum_{i=1}^n f_0(X_i) \big(\log f_0(X_i) - \log f_\bbeta(X_i)\big) + f_\bbeta(X_i)-f_0(X_i)\nonumber\\
&=&  \sum_{i=1}^n Y_i \big(\log f_0(X_i) - \log f_\bbeta(X_i)\big) + f_\bbeta(X_i)-f_0(X_i)\nonumber\\
&&\hspace{1cm} + \sum_{i=1}^n (f_0(X_i)-Y_i)\big(\log f_0(X_i) - \log f_\bbeta(X_i)\big)\nonumber\\
&=&\log{\mathcal L}(f_0)-\log{\mathcal L}(f_\bbeta) + \sum_{i=1}^n (f_0(X_i)-Y_i)\big(\log f_0(X_i) - \log f_\bbeta(X_i)\big).\nonumber
\end{eqnarray}
Therefore, for all $\bbeta\in\R^p$,
\begin{eqnarray*}
K(f_0, \widehat f^{gL}) - K(f_0, f_\bbeta) &=& l(\bbeta) - l(\widehat\bbeta^{gL}) + \sum_{i=1}^n \big(f_0(X_i)-Y_i\big)\big(\log f_\bbeta(X_i) - \log \widehat f^{gL}(X_i)\big) \nonumber\\
&=&   l(\bbeta) - l(\widehat\bbeta^{gL}) + \sum_{i=1}^n \big(f_0(X_i)-Y_i\big)\sum_{j=1}^p (\beta_j - \widehat\beta_j^{gL})\p_j(X_i)\nonumber\\
&=&  l(\bbeta) - l(\widehat\bbeta^{gL})+ \sum_{j=1}^p (\widehat\beta_j^{gL}-\beta_j)\sum_{i=1}^n \p_j(X_i)(Y_i - f_0(X_i)).\nonumber\\
\end{eqnarray*}
Let us write $\eta_j = \sum_{i=1}^n \p_j(X_i)(Y_i - f_0(X_i)) = \bA_j^T(\bY-\E[\bY]).$ We have
\begin{equation}\label{ineg_preuve}
K(f_0, \widehat f^{gL}) = K(f_0, f_\bbeta) + l(\bbeta) - l(\widehat\bbeta^{gL})+ (\widehat\bbeta^{gL} - \bbeta)^T\etab.
\end{equation}
By definition of $\widehat\bbeta^{gL}$,
$$-l(\widehat\bbeta^{gL}) + \sum_{k=1}^K \lambda_k^g\|\widehat\bbeta^{gL}_{G_k}\|_2 \leq -l(\bbeta) + \sum_{k=1}^K \lambda_k^g\|\bbeta_{G_k}\|_2.$$
Furthermore, on $\Omega_g$,
\begin{eqnarray}
|(\widehat\bbeta^{gL}-\bbeta)^T\etab| &=& \Big|\sum_{j=1}^p (\widehat\beta_j^{gL} - \beta_j)(\bA_j^T(\bY-\E \bY))\Big|\nonumber\\
&\leq& \sum_{k=1}^K \sum_{j \in G_k} |\widehat\beta_j^{gL} - \beta_j||\bA_j^T(\bY-\E \bY)|\nonumber\\
&\leq& \sum_{k=1}^K \Big(\sum_{j \in G_k} (\widehat\beta_j^{gL} - \beta_j)^2\Big)^{1/2}\Big(\sum_{j \in G_k} (\bA_j^T(\bY-\E \bY))^2\Big)^{1/2}\nonumber\\
&=& \sum_{k=1}^K \|\widehat\bbeta_{G_k}^{gL}-\bbeta_{G_k}\|_2\|\bA_{G_k}^T(\bY-\E \bY)\|_2\nonumber\\
&\leq& \sum_{k=1}^K \lambda_k^g \|\widehat\bbeta_{G_k}^{gL}-\bbeta_{G_k}\|_2. \label{ineg_eta}
\end{eqnarray}
Therefore, for all $\bbeta\in\R^p$,
$$K(f_0,\widehat f^{gL}) \leq K(f_0, f_\bbeta) + \sum_{k=1}^K \lambda_k^g\Big(\|\widehat\bbeta_{G_k}^{gL} - \bbeta_{G_k}\|_2 - \|\widehat\bbeta_{G_k}^{gL}\|_2 + \|\bbeta_{G_k}\|_2  \Big),$$
from which we deduce (\ref{slow_oi_gLasso}). \hfill$\square$
\subsection{Proof of Theorem \ref{fast_io_gl}}
We start from Equality (\ref{ineg_preuve}) combined with Inequality \eqref{ineg_eta}. 
Then, we have that on $\Omega_g$, for any $\bbeta$,
$$K(f_0,\widehat f^{gL}) + (\alpha-1)\sum_{k=1}^K \lambda_k^g \|\widehat\bbeta_{G_k}^{gL}-\bbeta_{G_k}\|_2 \leq K(f_0,f_\bbeta) + \sum_{k=1}^K \alpha\lambda_k^g\Big(\|\widehat\bbeta_{G_k}^{gL}-\bbeta_{G_k}\|_2 - \|\widehat\bbeta_{G_k}^{gL}\|_2 + \|\bbeta_{G_k}\|_2 \Big).$$
On $J(\bbeta)^c$, $\|\widehat\bbeta_{G_k}^{gL}-\bbeta_{G_k}\|_2 - \|\widehat\bbeta_{G_k}^{gL}\|_2 + \|\bbeta_{G_k}\|_2 =0$ and
\begin{equation}\label{no_25}
K(f_0,\widehat f^{gL}) + (\alpha-1)\sum_{k=1}^K \lambda_k^g \|\widehat\bbeta_{G_k}^{gL}-\bbeta_{G_k}\|_2 \leq K(f_0,f_\bbeta) + 2\alpha\sum_{k\in J(\bbeta)} \lambda_k^g\|\widehat\bbeta_{G_k}^{gL}-\bbeta_{G_k}\|_2.
\end{equation}
By applying the Cauchy-Schwarz inequality we also have
\begin{equation}\label{no_26}
K(f_0,\widehat f^{gL}) + (\alpha-1)\sum_{k=1}^K \lambda_k^g \|\widehat\bbeta_{G_k}^{gL}-\bbeta_{G_k}\|_2 \leq K(f_0,f_\bbeta) + 2\alpha |J(\bbeta)|^{1/2} \Big( \sum_{k\in J(\bbeta)} (\lambda_k^g)^2 \|\widehat\bbeta_{G_k}^{gL}-\bbeta_{G_k}\|_2^2 \Big)^{1/2}.
\end{equation}
If we write $\bDelta = \bD(\widehat\bbeta^{gL}-\bbeta)$, where $\bD$ is a diagonal matrix with $D_{j,j}=\lambda_k^g$ if $j \in G_k$, then we can rewrite (\ref{no_25}) as
\begin{equation}\label{no_27}
K(f_0,\widehat f^{gL}) + (\alpha-1)\|\bDelta\|_{1,2} \leq K(f_0,f_\bbeta) + 2\alpha\|\bDelta_{J(\bbeta)}\|_{1,2}
\end{equation}
and we deduce from (\ref{no_26})
\begin{equation}\label{no_28}
K(f_0,\widehat f^{gL}) \leq K(f_0,f_\bbeta) + 2\alpha(|J(\bbeta)|)^{1/2} \|\bDelta_{J(\bbeta)}\|_2.
\end{equation}
Now, on the event $\big\{2\alpha\|\bDelta_{J(\bbeta)}\|_{1,2} \leq \e K(f_0,f_\bbeta) \big\}$, the theorem follows immediately from (\ref{no_27}). We now assume that $\e K(f_0,f_\bbeta) \leq 2\alpha\|\bDelta_{J(\bbeta)}\|_{1,2} $. Since $K$ is non-negative, we deduce from (\ref{no_27}) that
$$(\alpha-1)\|\bDelta\|_{1,2} \leq 2\alpha\Big(1+{1 \over \e}\Big)\|\bDelta_{J(\bbeta)}\|_{1,2},$$
$$(\alpha-1)\|\bDelta_{J(\bbeta)^c}\|_{1,2} \leq \Bigg(2\alpha\Big(1+{1 \over \e}\Big)-(\alpha-1)\Bigg)\|\bDelta_{J(\bbeta)}\|_{1,2}$$
and
$$\|\bDelta_{J(\bbeta)^c}\|_{1,2} \leq\Bigg(\frac{\alpha+1+2\alpha/\e}{\alpha-1}\Bigg)\|\bDelta_{J(\bbeta)}\|_{1,2}.$$
This yields the following inequality for the vector $\bD^{-1}\bDelta= (\widehat\bbeta^{gL}-\bbeta)$: 
$$\|(\widehat\bbeta^{gL}-\bbeta)_{J(\bbeta)^c}\|_{1,2} \leq{\max_k \lambda_k^g \over \min_k \lambda_k^g}\frac{\alpha+1+2\alpha/\e}{\alpha-1}\|(\widehat\bbeta^{gL}-\bbeta)_{J(\bbeta)}\|_{1,2}. $$
From Assumption 2 we have that, if $\bbeta$ is such that $|J(\bbeta)|\leq s$, then 
$$\|(\widehat\bbeta^{gL}-\bbeta)_{J(\bbeta)}\|_2 \leq {1\over\kappa_n}\big((\widehat\bbeta^{gL}-\bbeta)^T\bG(\widehat\bbeta^{gL}-\bbeta)\big)^{1/2}.$$
Since
$$\bG_{j,j'}= \sum_{i=1}^n \p_j(X_i)\p_{j'}(X_i)f_0(X_i),$$
by setting
$$u_i=\log f_\bbeta(X_i) - \log f_0(X_i)\quad\mbox{and}\quad \widehat u_i^{gL}= \log f_{\widehat\bbeta^{gL}}(X_i) - \log f_0(X_i),$$
we have
\begin{eqnarray*}
(\widehat\bbeta^{gL}-\bbeta)^T\bG(\widehat\bbeta^{gL}-\bbeta) &=& \sum_{j=1}^p \sum_{j'=1}^p(\widehat\beta_j^{gL}-\beta_j)(\widehat\beta_{j'}^{gL}-\beta_{j'})\bG_{j,j'}\\
&=& \sum_{i=1}^n f_0(X_i)\Big(\sum_{j=1}^p (\widehat\beta_j^{gL}-\beta_j) \p_j(X_i)\Big)^2\\
&=& \sum_{i=1}^n f_0(X_i)(\widehat u_i^{gL} - u_i)^2.
\end{eqnarray*}
We set $h(f_0,f_\bbeta)=\sum_{i=1}^n f_0(X_i)u_i^2$ and  $h(f_0,\widehat f^{gL})=\sum_{i=1}^n f_0(X_i)(\widehat u_i^{gL})^2$. From (\ref{no_28}) and since
\begin{eqnarray*}
\|\bDelta_{J(\bbeta)}\|_2 &\leq& (\max_k \lambda_k^g) \|(\widehat\bbeta^{gL}-\bbeta)_{J(\bbeta)}\|_2\\
&\leq& {\max_k \lambda_k^g \over \kappa_n} \big((\widehat\bbeta^{gL}-\bbeta)^T\bG(\widehat\bbeta^{gL}-\bbeta)\big)^{1/2},
\end{eqnarray*}
we have
$$K(f_0,\widehat f^{gL}) \leq K(f_0,f_\bbeta) + {2\alpha\over\kappa_n}|J(\bbeta)|^{1/2}(\max_k \lambda_k^g)\Big(\sqrt{h(f_0,\widehat f^{gL})}+ \sqrt{h(f_0,f_\bbeta)}\Big).$$
 To conclude, we use arguments similar to \cite{Lemler}. We recall them for the safe of completeness. To connect $h(f_0,f_\bbeta)$ to $K(f_0,f_\bbeta)$, we use Lemma 1 of \cite{bach} that is recalled now.
\begin{lem}\label{lemme_bach}
Let $g$ be a convex three times differentiable function $g : \R \to \R$ such that for all $t \in \R$, $|g'''(t)|\leq Sg''(t)$ for some $S\geq 0$. Then, for all $t\geq 0$,
$${g''(0) \over S^2} \phi(-St)\leq g(t)-g(0)-g'(0)t \leq {g''(0) \over S^2} \phi(St),$$
where $\phi(x) = e^x - x - 1$.
\end{lem}
Let $h$ be a real function. We set
$$G(h) = \sum_{i=1}^n \big(e^{h(X_i)}-f_0(X_i)h(X_i)\big)$$
and
$$g(t) = G(h+tk),$$
where $h$ and $k$ are functions and $t \in \R$. We have :
$$
g'(t)=\sum_{i=1}^n \big( k(X_i)e^{h(X_i)+tk(X_i)} - f_0(X_i)k(X_i)\big),$$
$$g''(t)= \sum_{i=1}^n \big( k^2(X_i)e^{h(X_i)+tk(X_i)}\big)$$
and
$$g'''(t)= \sum_{i=1}^n \big( k^3(X_i)e^{h(X_i)+tk(X_i)}\big).$$
Therefore $|g'''(t)| \leq Sg''(t)$ with $S=\max_i |k(X_i)|$.
We choose $h(X_i) = \log f_0(X_i)$ and $k(X_i) = u_i = \log f_\bbeta(X_i) - \log f_0(X_i)$ and we apply Lemma \ref{lemme_bach} to $g$ with $t=1$. Computations yield that $g(1) - g(0) = K(f_0,f_\bbeta)$, $g'(0) = 0$ and $g''(0) = \sum_{i=1}^n f_0(X_i) u_i^2 = h(f_0,f_\bbeta)$.
Therefore
$${\phi(-S) \over S^2} h(f_0,f_\bbeta) \leq K(f_0,f_\bbeta) \leq {\phi(S)\over S^2}h(f_0,f_\bbeta).$$
Finally, using Assumption 1, for $\bbeta \in \Gamma(\mu)$, $S=\max_i |u_i|\leq \mu$. Furthermore, $x\longrightarrow{\phi(x) \over x^2}$ is a nonnegative increasing function and therefore we have
$$\mu' h(f_0,f_\bbeta) \leq K(f_0,f_\bbeta) \leq \mu'' h(f_0,f_\bbeta),$$
where $\mu' = {\phi(-\mu)\over \mu^2}$ and $\mu'' = {\phi(\mu) \over \mu^2}.$
It follows that, for $\bbeta \in \Gamma(\mu)$,
$$K(f_0,\widehat f^{gL}) \leq K(f_0,f_\bbeta) + {2\alpha\over\kappa_n\sqrt{\mu'}}|J(\bbeta)|^{1/2}(\max_k \lambda_k^g)\Big(\sqrt{K(f_0,\widehat f^{gL})}+ \sqrt{K(f_0,f_\bbeta)}\Big).$$
We use twice the inequality $2uv \leq bu^2+{v^2 \over b}$ for any $b>0$, applied to $u={\alpha\over\kappa_n}\sqrt{|J(\bbeta)|}(\max_k \lambda_k^g)$ and $v$ being either $\sqrt{{1\over\mu'}K(f_0,\widehat f^{gL})}$ or $\sqrt{{1\over\mu'}K(f_0,f_\bbeta)}$.
We have
$$\Big(1-{1\over\mu' b}\Big)K(f_0,\widehat f^{gL}) \leq \Big(1+{1\over\mu' b}\Big)K(f_0,f_\bbeta) + 2b{\alpha^2|J(\bbeta)|\over\kappa_n^2}(\max_k \lambda_k^g)^2.$$
Finally,
$$K(f_0,\widehat f^{gL}) \leq \Big({\mu'b+1\over\mu' b -1}\Big)K(f_0,f_\bbeta) + 2{\mu'b^2\over \mu'b-1}{\alpha^2|J(\bbeta)|\over\kappa_n^2}(\max_k \lambda_k^g)^2.$$
We choose $b>1/\mu'$ such that ${\mu'b+1\over\mu'b-1} = 1+\e$ and we set $B(\e,\mu) = 2(1+\e)^{-1}{\mu'b^2\over \mu'b-1}$. Finally, we have, for any $\bbeta \in \Gamma(\mu)$ such that $|J(\bbeta)|\leq s$,
$$K(f_0,\widehat f^{gL}) \leq (1+\e)\Bigg(K(f_0,f_\bbeta) + B(\e,\mu){\alpha^2|J(\bbeta)|\over\kappa_n^2}(\max_k \lambda_k^g)^2\Bigg).$$
This completes the proof of Theorem \ref{fast_io_gl}.
 \hfill$\square$

\bigskip

\noindent{\bf Acknowledgements:} The research of St\'ephane Ivanoff  and Vincent Rivoirard is partly supported by the french Agence Nationale de la Recherche (ANR 2011 BS01 010 01 projet Calibration). The research of Franck Picard is partly supported by the ABS4NGS ANR project ANR-11-BINF-0001-06. 

\bibliographystyle{apalike}

\end{document}